\documentclass{scrartcl}

\usepackage{amsmath}        
\usepackage{amsthm}         
\usepackage{amssymb}        
\usepackage{amsfonts}       
\usepackage{mathtools}      
\usepackage{bm}             
\usepackage{graphicx}       
\usepackage{subfig}         
\usepackage{algorithm}      
\usepackage{algpseudocode}  
\usepackage{booktabs}       
\usepackage{multirow}       
\usepackage{cite}           
\usepackage{hyperref}       
\usepackage{url}            
\usepackage{xcolor}         
\usepackage{balance}        
\usepackage{array}          
\usepackage{float}          
\usepackage{textcomp}       
\usepackage[margin=1.25in]{geometry}  

\newcommand{\euro}{\texteuro}
\newcommand{\surede}{\textsc{Surde}}
\newcommand{\awve}{\textsc{Awve}}
\newcommand{\ici}{\textsc{Ici}}
\newcommand{\sure}{\textsc{Sure}}
\newcommand{\mse}{\mathrm{MSE}}

\newcommand{\E}{\mathbb{E}}
\newcommand{\Var}{\mathrm{Var}}

\DeclareMathOperator*{\argmin}{arg\,min}

\theoremstyle{plain}      
\newtheorem{theorem}{Theorem}
\newtheorem{lemma}{Lemma}

\theoremstyle{definition}  
\newtheorem{definition}{Definition}

\theoremstyle{remark}      
\newtheorem{remark}{Remark}

\makeatletter
\def\blfootnote{\gdef\@thefnmark{}\@footnotetext}
\makeatother

\title{Adaptive Derivative Estimation via Stein's Unbiased Risk}

\author{Yonathan~Murin and Ali~Ozer~Ercan%
\thanks{The authors are with Meta Reality Labs (e-mail: moriny@meta.com; aliercan@meta.com).}\vspace{-4mm}} 

\date{}

\begin{document}
\maketitle

\blfootnote{This work has been submitted to the IEEE for possible publication. Copyright may be transferred without notice, after which this version may no longer be accessible.}

\begin{abstract}
Estimating derivatives from noisy sampled data is fundamental to         
control, human--computer interaction, and biomedical engineering.                     
Causal FIR derivative filters offer a natural approach for this challenge,
yet their performance depend on their length. While short filters amplify                   
noise, long filters introduce smoothing bias.                                     
We present \surede{} (SURE Derivative Estimator), which addresses this                  
tradeoff at each time step by evaluating a data-driven cost derived
from Stein's Unbiased Risk Estimator (SURE) across a bank of candidate
lengths and soft-combining their outputs via exponential weighting.
We prove a minimax-optimal oracle inequality for the soft-combined
estimator and use it to derive the optimal weighting temperature in
closed form. Thus, the only tuning parameter for  \surede{} is the
noise variance.
Via numerical simulations we show that \surede{} consistently outperforms 
alternative adaptive methods (the Intersection of Confidence Intervals (ICI) rule and the Adaptive
Windowing Velocity Estimator (AWVE)) for first-derivative estimation.
We further show that \surede{} is robust to noise-variance misspecification (9\%
degradation over a $4\times$ range), and that it is superior to ICI and AWVE also over real data scenarios (the EuRoC MAV dataset).
\surede{} is causal, computationally light, and requires only a rough
estimate of the noise variance.
\end{abstract}



\vspace{-4mm}
\section{Introduction}
\label{sec:introduction}

Estimating the derivative of a signal from discrete,
noisy samples is a common problem in many engineering disciplines.
In control systems, reliable velocity and acceleration estimates are
needed for feedback and feed-forward compensation
\cite{levant1998robust}.
In human--computer interaction, cursor velocity impacts both latency and
precision of pointing devices, as human sensitivity to visual feedback perturbations 
exhibits a velocity-dependent decay \cite{saunders2004}, 
and adaptive smoothing filters such as the 1\euro{} filter \cite{casiez2012oneeuro} are ubiquitous.
In brain--computer interfaces (BCIs), neural decoding pipelines estimate
the time derivatives of intended movement trajectories from noisy
cortical recordings \cite{schwartz2006bcireview}.
In each of these settings the estimator must operate \emph{causally}
and \emph{sequentially} to produce an estimate at every time step
using only \emph{current and past} observations. Additionally, it must cope with the
fundamental tradeoff between \emph{variance} (noise amplification) and
\emph{bias} (signal distortion).

The simplest causal derivative estimator of the signal $y_n$, 
sampled at sampling interval $T_s$, is the first-order difference $(y_n - y_{n-1})/T_s$.  
This estimator is unbiased for linear signals, yet it amplifies observation noise by a factor proportional to
$1/T_s$, leading to estimates with large variance at typical sampling rates. 
Increasing the length of the estimation window reduces the noise variance at the expense
of introducing bias: the filter implicitly fits a low-order polynomial
to the data, thus distorting signal components deviating from this polynomial model.
This tradeoff is governed by a single design parameter: the \emph{filter length} $N$.
In practice, it is common to choose $N$  by trial and error or by tuning over a
representative calibration signal. 
Such a fixed choice is inherently suboptimal because the optimal $N$
depends on the \emph{local} signal-to-noise ratio, which may vary over time.
This motivates an adaptive selection of $N$.

\vspace{-2mm}
\subsection{Related Work}
\label{ssec:intro_related}

\subsubsection{Finite  impulse response (FIR) Derivative Filters}
There are two main approaches for numerical differentiation. 
The first applies smoothing via window function (e.g., Boxcar, Hanning, Hamming)
to the coefficients of the difference filter, resulting in a filter
with well-understood spectral properties \cite{rabiner1970firdiff}.
The second fits a \emph{local} polynomial to the data using
\emph{least-squares} regression (or, equivalently, convolution with
Savitzky--Golay coefficients \cite{savitzky1964smoothing}) and then differentiates
the polynomial analytically. 
One can compare these approaches using the Variance Reduction Factor (VRF) framework
\cite{blair1997vrf} which measures the reduction to the noise level at the filter's 
output.

\subsubsection{Adaptive Bandwidth Selection}

Adaptive bandwidth selection in nonparametric statistics is rooted in Lepski’s 
principle~\cite{lepski1990problem}, which selects the largest window where estimates 
remain statistically consistent across resolutions. 
This approach was later extended to kernel regression to achieve minimax-optimal 
oracle inequalities~\cite{lepski1997optimal}. 
Katkovnik’s Intersection of Confidence Intervals (ICI) rule~\cite{katkovnik1999ici} 
provides a practical implementation of this theory, sequentially testing for overlapping 
confidence intervals at increasing scales. 
Balancing theoretical guarantees with computational efficiency, 
the ICI rule is now common in signal and image processing~\cite{katkovnik2006local}.

Based on similar principles, the Adaptive Windowing Velocity Estimator (AWVE)~\cite{sharifi2000awve} expands 
a backward window while maintaining polynomial-fit residuals below a noise-calibrated threshold. 
As detailed in Section~\ref{sec:adaptive_ci}, AWVE effectively functions as a variant of the ICI rule, 
substituting estimate consistency with a residual-based goodness-of-fit criterion.

These adaptive methods address the limitations of traditional estimators. 
While Kalman filtering~\cite{kalman1960newapproach} is optimal for known state-space models, 
it struggles with the model mismatch inherent in derivative estimation. 
Alternatively, total-variation (TV) regularization~\cite{chartrand2011tvdiff} excels at 
recovering piecewise-constant derivatives but is non-causal and computationally expensive. 
Similarly, Savitzky–Golay filters~\cite{savitzky1964smoothing} rely on non-causal symmetric windows, 
whereas their causal variants sacrifice data efficiency and incur boundary distortions.

\subsubsection{Stein's Unbiased Risk Estimator}

Stein's Unbiased Risk Estimator (\sure{}) \cite{stein1981sure} 
provides a tool for model selection: it forms an \emph{unbiased} 
estimate of the MSE of a given estimator \emph{without} access to the 
true signal, requiring only that the noise be Gaussian with known variance (or more generally, from an exponential family). 
\sure{} has been applied extensively for model selection in wavelet denoising \cite{donoho1995wavelet}, 
image restoration \cite{eldar2009gsure, ramani2008montecarlo}, and regularization parameter selection \cite{deledalle2014sugar}. 
Beyond these, it has been utilized for adaptive window-size selection in local polynomial regression \cite{Krishnan2012}, 
yet to the best of our knowledge, it has not previously been used for adaptive derivative filter design.

\vspace{-4mm}
\subsection{Main Contributions}
\label{ssec:intro_contributions}

In this work we make the following contributions:
                                                                       
\paragraph{SURE-based adaptive length selection}                                
  We derive a closed-form, per-sample cost function from \sure{} that  
  serves as an unbiased estimate of the MSE of \emph{any} linear FIR
  derivative filter.  Using this cost we propose \surede{}, an
  algorithm that evaluates a bank of candidate filter lengths at each
  time step and selects the one minimizing the estimated MSE
  (\emph{hard combining}), or forms a weighted average over all
  candidates via exponential weights (\emph{soft combining}).

\paragraph{Oracle inequality and optimal temperature}
  We prove that the soft-combined estimator satisfies a minimax-optimal
  oracle inequality, with excess risk
  $O(\sigma^2\sqrt{\log K})$ for $K$ candidate lengths.
  Minimizing this bound yields a closed-form expression for the
  soft-combining temperature in terms of the cost-fluctuation
  variance, eliminating the only free parameter and making
  \surede{} fully determined by the noise variance alone.

\paragraph{Empirical validation}
  Monte~Carlo simulations on synthetic signals show that \surede{}
  with optimal soft combining outperforms the ICI
  rule and \awve{} for first-derivative estimation. The method degrades by
  only 9\% under $4\times$ noise-variance misspecification, 
  compared to 36\% for \ici{} and 222\% for \awve{}.
  Experiments on real MAV trajectories (EuRoC dataset) confirm
  that \surede{} achieves the lowest error among adaptive methods
  at practical noise levels without retuning.

\vspace{-4mm}
\subsection{Organization}
\label{ssec:intro_organization}
\vspace{-1mm}

The rest of the paper is organized as follows.
Section~\ref{sec:preliminaries} provides the problem formulation and preliminaries, 
Section~\ref{sec:surede} introduces the \surede{} framework and Section~\ref{sec:oracle_inequality} 
derives an oracle inequality for \surede{}. 
Section~\ref{sec:numerical_results} reports simulation results
, while Section~\ref{sec:real_data} validates the method on real trajectory
data from the EuRoC MAV dataset.
Section~\ref{sec:conclusion} concludes the paper while discussing future
directions.


\vspace{-3mm}
\section{Problem Formulation and Preliminaries}
\label{sec:preliminaries}

\subsection{Notation}
\label{sec:notation}

We denote scalar signals at index $k$ as $x_k$ and their $n$-th order derivatives as $x_k^{(n)}$. 
Boldface lower-case letters represent column vectors (e.g., $\mathbf{x} = [x_0, \dots, x_{N-1}]^T$), 
while sans-serif upper-case letters denote matrices ($\mathsf{A}, \mathsf{B}, \mathsf{Q}$). 
Specifically, $\mathsf{I}_M$ and $\mathbf{1}_M$ are the identity matrix and all-ones vector of size $M$, respectively (omitted when context-evident). 
$\operatorname{Tr}(\cdot)$ denotes the matrix trace, $\mathcal{N}$ represents sets, 
and $\mathbb{E}\{\cdot\}$ is the expectation operator. 
For an FIR filter with coefficients $\mathbf{h}$, the $Z$-transform is $H(z) = \sum_{m} h_m z^{-m}$; 
linear convolution $\tilde{\mathbf{x}} = \mathbf{x} * \mathbf{h}$ follows the convention $\tilde{x}_k = \sum_{m} h_m x_{k-m}$.

\vspace{-4mm}
\subsection{Signal Model}
\label{sec:signal_model}

The observation model is given by:
\vspace{-2mm}
\begin{equation}
  y_k = x_k + w_k, \qquad k = 0, 1, 2, \dots
  \label{eq:observation_model}
\end{equation}
where $x_k$ is the unknown deterministic signal of interest and
$\{w_k\}$ is a zero-mean {\em wide-sense-stationary} noise process with known auto-correlation
function $r_w(\ell) = \mathbb{E}\{w_k \, w_{k-\ell}\}$.
Over a window of $N$ consecutive samples the observation vector is given by
$\mathbf{y} = \mathbf{x} + \mathbf{w}$,
with noise covariance matrix $\mathsf{Q}_w^{(N)} \hspace{-2pt}
  \hspace{-2pt} \in \hspace{-2pt} \mathbb{R}^{N \times N},\;
  [\mathsf{Q}_w^{(N)}]_{i,j} \hspace{-1pt} = \hspace{-1pt} r_w(i \hspace{-1pt} - \hspace{-1pt} j)$.

Our \textbf{objective} is to estimate the $n$-th derivative $x^{(n)}_k$ from the
noisy observations~$\{y_j\}_{j \le k}$ in a causal, sample-by-sample
manner.

\begin{definition}[Variance Reduction Factor (VRF)]
  \label{def:vrf}
  For a filter $\mathbf{h} \hspace{-1pt} = \hspace{-1pt} [h_0, h_1, \dots, h_{L-1}]^T$ applied to 
  white noise with variance $\sigma^2$, the output variance is 
  $\sigma_{\mathrm{out}}^2 \hspace{-1pt} = \hspace{-1pt} \sigma^2 \lVert \mathbf{h} \rVert_2^2$. and the
  VRF is $\mathrm{VRF}(\mathbf{h}) = \lVert\mathbf{h}\rVert_2^2$.
\end{definition}

\begin{remark}
    For colored noise, the output variance generalizes to 
    $\sigma_{\mathrm{out}}^2 = \mathbf{h}^T \mathsf{Q}_w^{(L)} \mathbf{h}$. 
    The VRF thus serves as a noise-independent metric for comparing 
    the noise suppression of filters with varying lengths.
\end{remark}

\vspace{-4mm}
\subsection{Derivative Filters}
\label{sec:derivative_filters}

\subsubsection{Direct Derivative Filters}
\label{sec:direct_filters}
The simplest estimate of the first derivative is the backward
difference $x^{(1)}_k \hspace{-1pt} = \hspace{-1pt} x_k \hspace{-1pt} - \hspace{-1pt} x_{k-1}$,
realized by the FIR filter $\mathbf{h}_d^{(1)}  = [1, -1]$ with
transfer function $H_d^{(1)}(z) \hspace{-1pt} = \hspace{-1pt} 1 \hspace{-1pt} - \hspace{-1pt} z^{-1}$.
Higher-order derivatives are obtained by repeated application of the
first-order difference, where the $n$-th order derivative filter has
transfer function $H_d^{(n)}(z) \hspace{-1pt} = \hspace{-1pt} \frac{1}{n!}\bigl(1 \hspace{-1pt} - \hspace{-1pt} z^{-1}\bigr)^n$,
corresponding to a length-$(n+1)$ FIR filter whose coefficients are
the binomial weights scaled by $1/n!$. This normalization 
ensures that the filter output approximates the $n$-th derivative 
(rather than the $n$-th finite difference) for signals sampled at unit rate.  
When the sampling interval is $\Delta t \hspace{-2pt} \neq \hspace{-2pt} 1$, an additional factor 
of $\Delta t^{-n}$ is needed.

For a window of $N$ observations the derivative filter can also be
represented in matrix form. Defining the Toeplitz matrix 
$\mathsf{H}_d^{(n)} \hspace{-1pt} \in \hspace{-1pt} \mathbb{R}^{(N-n) \times N}$ 
whose $i$-th row contains the coefficients of
$\mathbf{h}_d^{(n)}$ starting at column~$i$, then the vector of
noiseless derivative samples over the window is given by
$\mathbf{x}^{(n)} = \mathsf{H}_d^{(n)}\,\mathbf{x}$.

Applying the derivative filter directly to the noisy observations one obtains
\vspace{-1mm}
\begin{equation}
  \tilde{y}^{(n)}_k
  = (\mathbf{y} * \mathbf{h}_d^{(n)})_k
  = x^{(n)}_k + (\mathbf{w} * \mathbf{h}_d^{(n)})_k,
  \label{eq:noisy_derivative_raw}
\end{equation}
\vspace{-1mm}
which is an unbiased but typically very noisy estimate of
$x^{(n)}_k$.  To reduce the noise one may apply a second smoothing FIR filter
$\mathbf{h}_w(N)$ of length~$N$:
\vspace{-1mm}
\begin{equation}
  \hat{x}^{(n)}_k
  = \bigl(\mathbf{y} * \mathbf{h}_d^{(n)} * \mathbf{h}_w(N)\bigr)_k.
  \label{eq:derivative_estimate_cascade}
\end{equation}
As $N$ grows, the noise variance decreases but the estimate becomes more
biased if the signal deviates from the model assumed by
$\mathbf{h}_w$.  This fundamental \emph{bias--variance tradeoff} is
the central theme of Section~\ref{sec:bias_variance}.

\subsubsection{Least-Squares Derivative Filters}
\label{sec:ls_filters}

A more principled way to obtain the smoothing filter is to fit a local
polynomial model of degree $p \ge n$ to the most recent $N$ samples
and read off the $n$-th derivative from the fitted coefficients.
To formalize this approach, let $\mathbf{v} = [0, 1, \dots, N-1]^T$ be the 
vector of sample indices in the window, and define the centered index vector
$\tilde{\mathbf{v}} \hspace{-1pt} = \hspace{-1pt} \mathbf{v} \hspace{-1pt} - \hspace{-1pt} \mu_v$ with
$\mu_v \hspace{-1pt} = \hspace{-1pt} (N-1)/2$. The Vandermonde-like regression matrix is
\begin{equation}
  \mathsf{A}
  = \bigl[\mathbf{1},\;
          \tilde{\mathbf{v}},\;
          \tilde{\mathbf{v}}^{2},\;
          \dots,\;
          \tilde{\mathbf{v}}^{p}\bigr]
  \in \mathbb{R}^{N \times (p+1)},
  \label{eq:vandermonde}
\end{equation}
where $\tilde{\mathbf{v}}^{j}$ denotes the element-wise $j$-th
power. Given the observation window
$\tilde{\mathbf{x}} \hspace{-1pt} = \hspace{-1pt} [y_{k-N+1}, \dots, y_k]^T$, the coefficient
vector that minimizes the sum of squared residuals is
\begin{equation}
  \mathbf{a}
  = \bigl(\mathsf{A}^T \mathsf{A}\bigr)^{-1}
    \mathsf{A}^T \tilde{\mathbf{x}}
  = \mathsf{B}\,\tilde{\mathbf{x}},
  \label{eq:ls_solution}
\end{equation}

Under the local polynomial model of degree~$p$, the $n$-th
derivative at the \emph{right edge} of the window (index~$k$) is
given by the $(n+1)$-th row of $\mathsf{B}$:
\begin{equation}
\hat{x}^{(n)}_k
= \mathbf{b}_{n+1}^T \tilde{\mathbf{x}},
\qquad
\mathbf{b}_{n+1}^T
= [\mathsf{B}]_{n+1,:}\,.
\label{eq:ls_derivative_filter}
\end{equation}
The vector $\mathbf{v}(N) \triangleq \mathbf{b}_{n+1}$ (reversed)
is the FIR estimation filter of length~$N$.

\begin{remark}[Connection to Savitzky--Golay filters]
  \label{rem:sg}
  While the classical Savitzky--Golay filter~\cite{savitzky1964smoothing} uses a symmetric, 
  non-causal window, Eq.~\eqref{eq:ls_derivative_filter} employs a one-sided backward window, 
  ensuring causality for real-time processing. 
  Both methods evaluate the derivative at the window center, and for $p=1$
  this coincides with the right edge due to the constant slope.
\end{remark}

For a fixed degree $p$ and order $n$, the estimation filter $\mathbf{v}(N)$ depends 
solely on the window length $N \in \mathcal{N}$ and can be pre-computed. 
The derivative-smoothing operation in \eqref{eq:derivative_estimate_cascade} is compactly 
expressed as $\hat{x}^{(n)}_k \hspace{1pt} = \hspace{1pt} \mathbf{v}(N)^T \mathbf{y}_k(N)$, 
where $\mathbf{y}_k(N) \hspace{1pt} = \hspace{1pt} [y_{k-N+1}, \dots, y_k]^T$ 
denotes the observation window ending at time $k$.

\vspace{-3mm}
\subsection{The Bias--Variance Tradeoff}
\label{sec:bias_variance}

The MSE of the derivative estimate, calculated by treating $x_k^{(n)}$ as deterministic
and taking the expectation over the noise, decomposes into bias and variance:
\begin{equation}
    \mathrm{MSE}_{k}(N) = \underbrace{\bigl(\mathbb{E}\{\hat{x}^{(n)}_k\} - x^{(n)}_k\bigr)^2}_{\text{bias}^2} + \underbrace{\operatorname{Var}\{\hat{x}^{(n)}_k\}}_{\text{variance}}.
    \label{eq:mse_decomposition}
\end{equation}
This decomposition reveals a fundamental trade-off: 
short windows minimize bias but yield high variance, whereas long windows 
reduce variance at the expense of modeling bias. 
Since the optimal $N$ depends on the \emph{unknown}, time-varying signal, 
a data-driven criterion is required to adaptively balance these components at each sample $k$.

\vspace{-3mm}
\subsection{Adaptive Bandwidth via Confidence Intervals}
\label{sec:adaptive_ci}

A natural approach to the adaptive window selection problem is to
start with the shortest (highest-variance) window and grow it as
the resulting estimates remain statistically consistent. 
Two related algorithms rely this principle: the Intersection of
Confidence Intervals (ICI) rule~\cite{katkovnik1999ici} and the
Adaptive Windowing Velocity Estimator (\awve{})~\cite{awve}.

The \ici{} rule~\cite{katkovnik1999ici}, grounded in Lepski’s adaptive 
estimation theory~\cite{lepski1990problem, lepski1997optimal}, selects the 
largest window length statistically consistent with all smaller-scale estimates. 
Its sole tuning parameter, $\Gamma$, controls the confidence interval width: 
larger values favor longer, lower-variance windows, while smaller values 
prioritize shorter, lower-bias windows. 
With an appropriate $\Gamma$, the \ici{} rule achieves the near-optimal 
oracle inequalities established by Lepski~\cite{lepski1997optimal}.

The \awve{} provides a related but distinct heuristic based on residual testing. 
It starts with the shortest window and iteratively increases its length as long as the polynomial model adequately explains the
data (all residuals below $\alpha\,\sigma$). The first window whose residuals exceed the threshold signals a model mismatch,
and the previous window length is selected.

\begin{remark}[Connection between \awve{} and \ici{}]
  \label{rem:awve_ici_connection}
  \awve{} and \ici{} share a sequential structure, expanding candidate windows from 
  shortest to longest until a statistical test fails. 
  The fundamental distinction lies in their criteria: \ici{} evaluates estimate consistency 
  via the intersection of confidence intervals across bandwidths, 
  whereas \awve{} assesses model adequacy through an intra-bandwidth goodness-of-fit 
  test on polynomial-fit residuals.
    
  Both methods grow the window until significant bias is detected. 
  Hence, \awve{} functions as an \ici{} variant where the consistency test is replaced by a 
  residual-based model-mismatch check. Their tuning parameters, $\Gamma$ for \ici{} and $\alpha$ 
  for \awve{}, are analogous: larger values increase bias tolerance, favoring longer, 
  lower-variance windows.
\end{remark}


\vspace{-3mm}
\section{The SURE Derivative Estimator (\surede{})}
\label{sec:surede}

We now present our main contribution: a principled, 
computationally efficient alternative to \ici{} and \awve{} based on 
Stein’s Unbiased Risk Estimate (SURE). 
By replacing heuristic thresholds with an unbiased MSE estimate \eqref{eq:mse_decomposition}, 
this approach enables both hard selection and soft combination of window lengths. 

\vspace{-3mm}
\subsection{Hard-Combining \surede{}}
\label{sec:surede_hard}

Recalling that the candidate window lengths are $N \hspace{-2pt} \in \hspace{-2pt} \mathcal{N}$, we start with defining the following variables.
Let $N_0$ be the number of recent samples over which the $n$-th
derivative is approximately constant
\footnote{If the derivative is not approximately constant over some window then estimation 
via filtering is not relevant. In practice, 
this implies that the sampling rate is high enough w.r.t. the rate of change of the signal.}, i.e.\
$x^{(n)}_k \hspace{-1pt} \approx \hspace{-1pt} x^{(n)}_{k-1} \hspace{-1pt} \approx \hspace{-1pt} \cdots \hspace{-1pt} \approx \hspace{-1pt}
x^{(n)}_{k-N_0+1}$. In practice $N_0 \hspace{-1pt} = \hspace{-1pt} N_1 \hspace{-1pt} - \hspace{-1pt} n$ 
where $N_1 \hspace{-1pt} = \hspace{-1pt} \min(\mathcal{N})$.

Further, $\mathbf{s}^{(n)}_{N_0}$ is the \emph{s-vector} of length~$N$,
defined as the column-wise partial sums of the first $N_0$ rows of
the derivative filter matrix:
\begin{equation}
    \bigl(\mathbf{s}^{(n)}_{N_0}\bigr)^T
    = \bar{\mathbf{1}}_{N_0}^T \mathsf{H}_d^{(n)}
    \in \mathbb{R}^{1 \times N},
    \label{eq:s_vector_def}
\end{equation}
where $\bar{\mathbf{1}}_{N_0}
\hspace{-1pt} = \hspace{-1pt} [\mathbf{1}_{N_0};\, \mathbf{0}_{N-n-N_0}] \hspace{-1pt} \in \hspace{-1pt} \mathbb{R}^{N-n}$
selects the first $N_0$ rows of $\mathsf{H}_d^{(n)}$.

Finally, define the cost function $\mathfrak{c}_{k}(N)$, where in Appendix \ref{app:derivation} 
we prove that this is the \sure{} cost function:
\vspace{-3mm}
\begin{align}
  \mathfrak{c}_{k}(N)
   & = N_0 \bigl(\mathbf{v}(N)^T \mathbf{y}_k(N)\bigr)^2
  + 2\,\mathbf{v}(N)^T \mathsf{Q}_w^{(N)}
      \mathbf{s}^{(n)}_{N_0} \notag \\
   & \qquad - 2\,\mathbf{v}(N)^T \mathbf{y}_k(N)\;
      \bigl(\mathbf{s}^{(n)}_{N_0}\bigr)^T \mathbf{y}_k(N),
  \label{eq:sure_cost}
\end{align}  

The following theorem states that $\mathfrak{c}_{k}(N)$ is an unbiased estimate of $\mathrm{MSE}_{k}(N)$ up to a constant independent of $N$.

\begin{theorem}[Unbiasedness of the SURE cost]
  \label{thm:sure_unbiased}
  Under the signal model \eqref{eq:observation_model} with Gaussian
  noise $\mathbf{w} \sim \mathcal{N}(\mathbf{0}, \mathsf{Q}_w^{(N)})$,
  the cost function $\mathfrak{c}_{k}(N)$ satisfies
  \begin{equation}
    \mathbb{E}\bigl\{\mathfrak{c}_k(N)\bigr\}
    = \mathrm{MSE}_k(N) + C,
    \label{eq:sure_unbiased}
  \end{equation}
  where $C$ is independent of $N$. Consequently, minimizing $\mathfrak{c}_{k}(N)$ over $\mathcal{N}$ 
  is equivalent to minimizing the MSE up to $N$-invariant terms.
\end{theorem}

\begin{proof}
  See Appendix~\ref{app:derivation}.
\end{proof}
\vspace{-1mm}

Based on Theorem \ref{thm:sure_unbiased}, the hard-combining \surede{} selects the window length that minimizes the estimated risk 
$N_k^{\mathrm{hard}} \hspace{-1pt} = \hspace{-1pt} \operatorname*{argmin}_{N \in \mathcal{N}} \mathfrak{c}(N)$
The hard-combining \surede{} algorithm is summarized in Algorithm \ref{alg:surede_hard}.

\begin{algorithm}[t]
\caption{Hard-Combining \surede{}}
\label{alg:surede_hard}
\begin{algorithmic}[1]
\Require A set $\mathcal{N} \hspace{-2pt} = \hspace{-2pt} \{N_1 \hspace{-2pt} < \hspace{-2pt} \cdots \hspace{-2pt} < \hspace{-2pt} N_M\}$;
         polynomial degree~$p$; derivative order~$n$;
         noise covariance $\mathsf{Q}_w$; parameter~$N_0$.
\Statex  \textit{Pre-computation (once):}
\State For each $N \in \mathcal{N}$ compute and store $\mathbf{v}(N)$
       via \eqref{eq:ls_derivative_filter}.
\State For each $N \in \mathcal{N}$ compute and store
       $\mathbf{s}^{(n)}_{N_0}$ via \eqref{eq:s_vector_def} and
       $\tau(N) = \mathbf{v}(N)^T \mathsf{Q}_w^{(N)} \mathbf{s}^{(n)}_{N_0}$.
\Statex  \textit{On-line (every sample~$k$):}
\For{each $N \in \mathcal{N}$}
  \State $e(N) \hspace{-2pt} \gets \hspace{-2pt} \mathbf{v}(N)^T\,\mathbf{y}_k(N)$
          \Comment{LS estimate for window~$N$}
  \State $r(N) \hspace{-2pt} \gets \hspace{-2pt} \bigl(\mathbf{s}^{(n)}_{N_0}\bigr)^T
                      \mathbf{y}_k(N)$
  \State $\mathfrak{c}(N) \hspace{-2pt} \gets \hspace{-2pt} N_0 e(N)^2
          +  2\,\tau(N)
         \hspace{-2pt} - \hspace{-2pt} 2e(N) r(N)$
\EndFor
\State $N^* \gets \operatorname{argmin}_{N}\;\mathfrak{c}(N)$
\State \Return $\hat{x}^{(n)}_k = e(N^*)$
\end{algorithmic}
\end{algorithm}

\begin{remark}[Generalization to linear estimators]
  \label{rem:sure_any_linear}
  The \sure{} cost \eqref{eq:sure_cost} and subsequent oracle inequality (Theorem~\ref{thm:oracle_ineq}) 
  are not restricted to LS polynomial filters. 
  They apply to any linear estimator $\hat{x}^{(n)}_k = \mathbf{v}^T \mathbf{y}_k(N)$, 
  including Lanczos~\cite{lanczos1956applied}, noise-robust~\cite{holoborodko2008smooth}, 
  maxflat~\cite{selesnick2002maximally}, or windowed-difference designs. 
  The framework is filter-agnostic: provided the estimator is linear in $\mathbf{y}_k(N)$, 
  all theoretical guarantees hold with an effective variance $V(N) = \mathbf{v}^T \mathsf{Q}_w \mathbf{v}$ 
  tailored to the specific filter coefficients. 
  While we prioritize LS filters for their minimal variance among polynomial designs 
  (Section~\ref{sec:ls_filters}), any custom FIR derivative filter may be substituted.
\end{remark}

\vspace{-3mm}
\subsection{Computational Complexity Comparison}
\label{sec:complexity}

After pre-computation of the filter vectors $\mathbf{v}(N)$,
the s-vectors $\mathbf{s}^{(n)}_{N_0}$, and the scalar
$\mathbf{v}(N)^T \mathsf{Q}_w^{(N)} \mathbf{s}^{(n)}_{N_0}$,
the on-line cost of evaluating $\mathfrak{c}(N)$ for one window
length is $\approx7N$ flops. For \ici{} the complexity is $\approx 2N + O(1)$ 
with pre-computed $\mathbf{v}$, $\lVert\mathbf{v}\rVert_2$, while for \awve{} it is
$\approx 2N^2 + 3N$ where $\mathsf{B}$ matrices are pre-computed.
  
\surede{} reduces the per-window cost from $O(N^2)$ (\awve) to $O(N)$. 
While \ici{} also achieves $O(N)$ complexity, it relies on a heuristic parameter $\Gamma$ 
rather than direct MSE estimation. 
For a candidate set $\mathcal{N} \hspace{-2pt} = \hspace{-2pt} \{5, \dots, 80\}$, 
\awve{} requires $\approx 13,040$ flops per sample at $N=80$, 
whereas \surede{} requires only $\approx 560$—a $23\times$ reduction. 
Consequently, the total online cost for all \surede{} candidates is 
less than a \emph{single} \awve{} evaluation at the largest window.

\vspace{-3mm}
\subsection{Conceptual Comparison: \ici{}, \awve{}, and \surede{}}
\label{sec:conceptual_comparison}

The three adaptive methods (\ici{}, \awve{}, and \surede{}) all aim to optimize the bias-variance 
trade-off by selecting an appropriate window length $N$, but they employ fundamentally different mechanisms.

\ici{} and \awve{} are sequential testing methods that expand the window until a statistical criterion is violated. 
Specifically, \ici{} performs an inter-bandwidth consistency test~\cite{lepski1990problem, lepski1997optimal}, 
intersecting confidence intervals across different scales. 
It inherits oracle inequalities ensuring the selected risk is within $O(\log K)$ of the oracle risk, 
though its performance relies on the user-tuned parameter $\Gamma$. 
\awve{} instead utilizes an intra-bandwidth goodness-of-fit test, 
checking if a residual bound properly explains the local polynomial model via the threshold $\alpha$.

In contrast, \surede{} is a risk minimization method that directly estimates the MSE \eqref{eq:mse_decomposition} 
at each candidate $N$ using Stein’s unbiased risk estimate. 
Unlike the heuristic thresholds of \ici{} and \awve{}, \surede{} targets the actual MSE 
objective without requiring a selection-criterion threshold. 
This framework also naturally extends to soft combination (Sec.~\ref{sec:soft_combining}), 
which avoids discrete switching and yields smoother adaptation by weighting multiple window lengths.


\vspace{-2mm}
\section{Oracle Inequality for \sure{} Selection}
\label{sec:oracle_inequality}

A natural question about any data-driven model selection procedure is:
how close is it to the \emph{oracle} that knows the best model (length) in advance?  
In this section we establish a finite-sample oracle inequality for the hard-combining
\surede{} algorithm (Algorithm~\ref{alg:surede_hard}), showing that it
achieves near-oracle risk up to a term that grows only logarithmically
in the number of candidate window lengths.  We then compare this
result to the classical \ici{} oracle bound. 

\vspace{-3mm}
\subsection{Setup and Notation}
\label{sec:oracle_setup}

Recall the observation model \eqref{eq:observation_model}
For each candidate window length $ N \hspace{-2pt} \in \hspace{-2pt} \mathcal{N} \hspace{-2pt} = \hspace{-2pt} \{N_1, \ldots, N_K\}$,
the LS derivative filter $\mathbf{v}(N)$ produces the estimate
$\hat{x}^{(n)}_k(N) = \mathbf{v}(N)^T \mathbf{y}_k(N)$.
The \emph{risk} (expected MSE conditioned on the signal) at window
length~$N$ is then given by
\begin{equation}
  R(N) \hspace{-2pt}
  \triangleq \hspace{-2pt} \E\bigl\{
    \bigl(\hat{x}^{(n)}_k(N) \hspace{-2pt} - \hspace{-2pt} x^{(n)}_k\bigr)^2
  \bigr\}
  \hspace{-2pt} = \hspace{-2pt} b^2(N) \hspace{-2pt} + \hspace{-2pt} V(N),
  \label{eq:risk_def}
\end{equation}
where $b(N) = \mathbf{v}(N)^T \mathbf{x}_k(N) - x^{(n)}_k$ is the
bias and
$V_w(N) = \mathbf{v}(N)^T \mathsf{Q}_w^{(N)} \mathbf{v}(N)$ is the
variance.

The \emph{oracle} window length and risk are then given by
\vspace{-2mm}
\begin{equation}
  N^* = \argmin_{N \in \mathcal{N}}\; R(N), \qquad R^* \triangleq R(N^*).
  \label{eq:oracle_N}
\end{equation}
\vspace{-1mm}

The hard-combining \surede{} selects
$\hat{N} \hspace{-2pt} = \hspace{-2pt} \argmin_{N \in \mathcal{N}} \mathfrak{c}(N)$,
where $\mathfrak{c}(N)$ is the \sure{} cost
\eqref{eq:sure_cost}.  By Theorem~\ref{thm:sure_unbiased},
$\E\{\mathfrak{c}(N)\} = R(N) + C$ for a constant~$C$ independent
of~$N$.  We can therefore write
\begin{equation}
  \mathfrak{c}(N)
  = R(N) + C + \xi(N),
  \label{eq:sure_decomposition}
\end{equation}
where $\xi(N) \triangleq \mathfrak{c}(N) - \E\{\mathfrak{c}(N)\}$ is
a zero-mean random variable (RV) representing the fluctuation of the \sure{}
cost around the true risk.

\vspace{-3mm}
\subsection{An Oracle Inequality}
\label{sec:oracle_main}                                                                

A central concern in data-driven selection is whether adaptation might "follow the noise," 
selecting a window with low \sure{} cost but high true risk. 
The oracle inequality below addresses this by bounding the \surede{} risk relative to the oracle risk, 
$R^* = \min_{N \in \mathcal{N}} R(N)$, with an additive penalty scaling as $O(\sqrt{\log K})$. 
While any fixed window $N_j$ trivially satisfies an additive bound $R(N_j) = R^* + \Delta_j$, 
the boundary cases are particularly illustrative: under standard bias monotonicity, 
the shortest window satisfies $R(N_1) \le R^* + V(N_1)$ and the longest $R(N_K) \le R^* + b^2(N_K)$. 
The following theorem establishes that the adaptive \surede{} procedure similarly avoids catastrophic overfitting, 
maintaining a near-oracle risk across the entire candidate set.

\begin{theorem}[Oracle inequality for \surede{}]
\label{thm:oracle_ineq}
Let $\hat{N} = \arg\min_{N \in \mathcal{N}} \mathfrak{c}(N)$ be the window selected from $K$ candidates.
Under the Gaussian observation model \eqref{eq:observation_model},
the risk $R(\hat{N})$ satisfies $R(\hat{N}) \le R^* + 2\max_{N \in \mathcal{N}} |\xi(N)|$. 
Moreover, defining the effective noise variance via
\begin{equation}
    V_{\mathrm{eff}}(N)
    \triangleq \max\bigl(V_{w}(N),\;
      \mathbf{s}^T\mathsf{Q}_w\mathbf{s}\bigr),
    \label{eq:Veff_def}
\end{equation}
and $\overline{V} \hspace{-2pt} \triangleq \hspace{-2pt} \max_{N \in \mathcal{N}} V_{\mathrm{eff}}(N)$,
then with probability at least $1 - \delta$:
\begin{equation}
  R(\hat{N})
  \hspace{-2pt} \le \hspace{-2pt} R^*
    \hspace{-1pt} + \hspace{-1pt} C_3 \overline{V} \sqrt{\log(2K/\delta)}
    \hspace{-1pt} + \hspace{-1pt} C_4 \overline{V} \log(2K/\delta).
  \label{eq:oracle_ineq_prob}
\end{equation}
Here, $C_3, C_4 \hspace{-2pt} > \hspace{-2pt} 0$ are constants depending on $N_0$,
$\max_N \lVert\mathbf{v}(N)\rVert$, and the signal.
Finally the following in-expectation inequality holds:
\begin{equation}
  \E\bigl[R(\hat{N})\bigr]
  \le R^*
    + C_5\,\overline{V}\,
      \sqrt{\log K},
  \label{eq:oracle_ineq_expect}
\end{equation}
with $C_5$ depends on the same quantities as $C_3, C_4$.
\end{theorem}

  \begin{proof}
    The proof is provided in Appendix~\ref{app:oracle_proofs}.
  \end{proof}

\begin{remark}[Theoretical Guarantee vs. Practical Benefit]
  \label{rem:bound_vs_practice}
  The oracle inequality \eqref{eq:oracle_ineq_expect} is a "safety guarantee": 
  it ensures that adaptation does not significantly \emph{hurt} performance, 
  capping the risk at $O(\overline{V}\sqrt{\log K})$ above the oracle. 
  However, the true value of \surede{} is its ability to \emph{help} across shifting regimes. 
  While a fixed short window $N_1$ suffers from high variance in noisy periods 
  and a fixed long window $N_K$ incurs high bias during rapid signal changes, 
  \surede{} dynamically tracks the oracle $N^*$. 
  By adapting to both noise levels and local signal dynamics, 
  it maintains near-oracle risk where static strategies would otherwise fail.
\end{remark}

\begin{remark}[Minimax Optimality of the $\sqrt{\log K}$ Rate]
  \label{rem:sqrt_log_K_optimal}
  The $\sqrt{\log K}$ rate in \eqref{eq:oracle_ineq_expect} is minimax optimal for 
  SURE-based selection among $K$ linear estimators. 
  Cavalier and Golubev~\cite{cavalier2006risk} established that any unbiased risk 
  selection rule achieves an oracle inequality with an additive remainder 
  of $O(\overline{V}\sqrt{\log K})$, and that there exist signal configurations 
  where this penalty is unavoidable, yielding a minimax risk of $\Theta(\overline{V}\sqrt{\log K})$.

  The intuition for why the $\sqrt{\log K}$ penalty cannot be avoided
  stems from extreme value theory for Gaussian RVs, see
  \cite{leadbetter1983extremes} and \cite{boucheron2013concentration}.
  The \sure{} cost $\mathfrak{c}(N)$ deviates from the true risk
  $R(N)$ by a zero-mean fluctuation $\xi(N)$.  When selecting the
  minimizer over $K$ candidates, the selection is influenced by the
  \emph{maximum} of $K$ such fluctuations.  Classical results on
  the maximum of $K$ independent standard Gaussian RVs state that
  $\E[\max_{1 \le i \le K} Z_i] \hspace{-2pt} = \hspace{-2pt} \Theta(\sqrt{\log K})$. 
  Hence the ``winner's curse'', namely the tendency of the selected candidate to have
  a favorably low fluctuation, introduces an irreducible bias of
  order $\sqrt{\log K}$ in the risk estimate. 
  Fortunately, this logarithmic growth is extremely mild; even for 
  $K \hspace{-2pt} = \hspace{-2pt} 100$, $\sqrt{\log 100} \hspace{-2pt} \approx \hspace{-2pt} 2.15$. 
  This ensures practitioners can include a generous set of candidate window lengths 
  without incurring meaningful performance degradation.
\end{remark}

\vspace{-3mm}
\subsection{Comparison with the \ici{} Oracle Bound and \awve{}}
\label{sec:ici_comparison}

The \ici{} rule \cite{goldenshluger1997adaptive, katkovnik2006local} is a nonparametric 
bandwidth selector whose oracle inequality \cite{lepski1997optimal} takes the form:
\begin{equation}
  \E\bigl[R(\hat{N}_{\mathrm{ICI}})\bigr]
  \le C_{\mathrm{ICI}}\,
    \min_{N \in \mathcal{N}} R(N)
    \cdot \log K,
  \label{eq:ici_oracle}
\end{equation}
where $C_{\mathrm{ICI}} > 0$ is a constant depending on the
threshold parameter~$\Gamma$.

\begin{remark}[Additive vs. Multiplicative Oracle Bounds]
  \label{rem:additive_vs_mult}
  The structural difference between the \surede{} bound \eqref{eq:oracle_ineq_expect} 
  and the \ici{} bound \eqref{eq:ici_oracle} highlights their performance 
  across different operating regimes:
  \begin{itemize}
    \item \textbf{\surede{} (Additive):} $R(\hat{N}) \le R^* + O(\overline{V}\sqrt{\log K})$. 
    The penalty is a constant that does not scale with the oracle risk. 
    This is superior in bias-dominated regimes (large $R^*$), 
    where the adaptation overhead remains fixed.
    \item \textbf{\ici{} (Multiplicative):} $R(\hat{N}) \le C \cdot R^* \log K$. 
    The penalty is proportional to the oracle risk. While potentially tighter in noise-dominated regimes where $R^*$ is small, 
    this bound grows significantly as signal complexity increases.
  \end{itemize}
  In balanced regimes where bias $\approx$ variance, the two bounds are comparable. However, \surede{} 
  ensures that the ``price of adaptation'' does not inflate alongside the underlying signal risk.
\end{remark}

\begin{remark}[Comparison with Leave-One-Out Cross-Validation]
  \label{rem:loocv_comparison}
  While Leave-one-out cross-validation (LOOCV) achieves similar oracle guarantees under broad conditions \cite{li1986asymptotic}, 
  it requires recomputing the estimator for each sample, leading to a $O(K \hspace{-2pt} \cdot \hspace{-2pt} N_{\max}^2)$ per-step cost. 
  \surede{} provides a comparable guarantee with significantly lower $O(K \hspace{-2pt} \cdot \hspace{-2pt} N_{\max})$ complexity, 
  leveraging the closed-form \sure{} expression available for linear estimators.
\end{remark}

\vspace{-3mm}
\subsection{Soft-Combining \surede{}}                                                  
\label{sec:soft_combining}    

Hard selection (Algorithm~\ref{alg:surede_hard}) picks a single window 
$\hat{N} \hspace{-2pt} = \hspace{-2pt} \arg\min_{N} \mathfrak{c}(N)$ at each sample. 
This ``winner-takes-all'' approach can be sensitive to fluctuations, 
causing abrupt switching between candidates with similar costs. 
To achieve smoother adaptation, we introduce soft-combining, which forms a \emph{weighted average} of all candidate estimates:
\vspace{-1mm}
\begin{equation}
    \hat{x}^{(n)}_{\mathrm{soft}}
    = \sum_{N \in \mathcal{N}} w(N)\,\hat{x}^{(n)}_k(N),
\label{eq:soft_estimate}
\end{equation}
\vspace{-1mm}
The weights $w(N)$ are derived from the \sure{} costs, 
constrained such that $w(N) \hspace{-2pt} \ge \hspace{-2pt} 0$ and $\sum w(N) \hspace{-2pt} = \hspace{-2pt} 1$. 
Our previous oracle inequality provides the theoretical foundation for determining these optimal weights. 

Ideally, the weights~$w$ should be chosen to minimize the risk
of the combined estimator. As the \sure{} cost $\mathfrak{c}(N)$ provides 
an unbiased estimate of the per-window risk (up to a constant), a
natural objective is
\vspace{-1mm}
\begin{equation}
\min_{w \in \Delta_K} \hspace{-1pt}
  \sum_{N \in \mathcal{N}} \hspace{-1pt} w(N)\mathfrak{c}(N), 
\label{eq:linear_objective}
\end{equation}
where $\Delta_K$ is the probability simplex.  
Minimizing~\eqref{eq:linear_objective}
directly yields the hard-combining solution $w(\hat{N}) = 1$, 
which is the strategy we wish to improve upon. 
Hence, to obtain a smooth weighting, we add an entropy regularizer that penalizes concentration:
\vspace{-1mm}
\begin{equation}
\min_{w \in \Delta_K}\;
  \sum_{N \in \mathcal{N}} \hspace{-1pt} w(N) \mathfrak{c}(N)
  \hspace{-1pt} + \hspace{-1pt} T \hspace{-1pt} \sum_{N \in \mathcal{N}} \hspace{-1pt} w(N) \log w(N),
\label{eq:ewa_objective}
\end{equation}
where $T \hspace{-2pt} > \hspace{-2pt} 0$ is a temperature parameter.  
The negative of the entropy term is maximized by the uniform distribution and minimized by a point mass, 
hence acting as a counterforce to the cost-driven concentration. 
A Lagrangian analysis yields the closed-form solution:
\vspace{-1mm}
\begin{equation}
w(N) = \frac{\exp\bigl(-\mathfrak{c}(N) / T\bigr)}
            {\displaystyle\sum_{M \in \mathcal{N}}
             M \cdot \exp\bigl(-\mathfrak{c}(M) / T\bigr)}.
\label{eq:soft_weights}
\end{equation} 
This is the \emph{exponential weighted aggregate} (EWA) studied
in the model aggregation literature
\cite{catoni2004statistical, dalalyan2008aggregation, rigollet2012sparse}.
The temperature~$T$ interpolates between two extremes. 
For $T \hspace{-2pt} \to \hspace{-2pt} 0$ weights concentrate on $\hat{N} = \argmin_N \mathfrak{c}(N)$, 
recovering hard-combining. On the other hand, for $T \hspace{-2pt} \to \hspace{-2pt} \infty$ 
weights become proportional to the prior, ignoring the costs entirely.
$T$ controls the bias--variance tradeoff of the
\emph{combining strategy itself}.  A small~$T$ tracks the
per-sample costs closely but amplifies cost
noise; a large~$T$ smooths over noise
but dilutes the estimate by mixing in suboptimal windows.

To determine the temperature $T$, we apply the EWA oracle
inequality \cite{catoni2004statistical, dalalyan2008aggregation} to the \surede{} framework:
\begin{equation}
  \E\bigl[R_{\mathrm{soft}}\bigr] \le R^* + T \log K + \frac{\max_{N} \Var\bigl[\xi(N)\bigr]}{4T},
\label{eq:ewa_bound}
\end{equation}
where $R^*$ is the oracle risk and $\xi(N)$ is the cost fluctuation.
Minimizing this bound with respect to $T$ balances the regularization
penalty ($T \log K$) against the estimation penalty ($\Var[\xi]/4T$).
The resulting optimal temperature yields an oracle inequality with
the same $\sqrt{\log K}$ rate as the hard-combining case (Theorem~\ref{thm:oracle_ineq}):
\begin{equation}
\E\bigl[R_{\mathrm{soft}}\bigr] \le R^* + \sqrt{\max_{N} \Var\bigl[\xi(N)\bigr] \cdot \log K}.
\label{eq:ewa_oracle}
\end{equation}
The full variance $\Var[\xi]$ includes a signal-dependent component that is 
not computable without oracle knowledge. At the oracle window~$N^*$, however, 
bias and variance are approximately balanced. We therefore calibrate $T$ to 
this balanced regime yielding the following theorem.
\begin{theorem}[Balanced-regime optimal temperature]
    \label{thm:optimal_temp}
    Let $V_w$ and $S_w \hspace{-2pt} = \hspace{-2pt} \mathbf{s}^T \mathsf{Q}_w \mathbf{s}$ be evaluated
    at the shortest candidate window $N_1$. Further define $C_{vs} \hspace{-2pt} = \hspace{-2pt} \mathbf{v}^T \mathsf{Q}_w \mathbf{s}$
    and $\nu \hspace{-2pt} \triangleq \hspace{-2pt} 2N_0^2 V_w^2 \hspace{-1pt} - 
         \hspace{-1pt} 8N_0 V_w C_{vs} \hspace{-1pt} + \hspace{-1pt} 4V_w S_w \hspace{-1pt} + 
         \hspace{-1pt} 4C_{vs}^2$.
    The balanced-regime optimal temperature is given by:
    \begin{equation}
    T^* = \sqrt{\frac{\nu}{2\,\log K}},
    \label{eq:T_star_explicit}
    \end{equation}
\end{theorem}

\begin{proof}
    The proof is provided in Appendix \ref{app:opt_temp_proof}.
\end{proof}

The term under the square root in \eqref{eq:T_star_explicit} is a dimensionless constant                                                                         
determined by the derivative order $n$, degree $p$, and candidate set parameters.                                                                                
Computed once during initialization, it requires no tuning.                                                                                                      
For $K \hspace{-2pt} = \hspace{-2pt}6$ and $N_1 \hspace{-2pt}= \hspace{-2pt} 4$,                                                                                 
the ratio $T^{*}/\sigma^2$ drops from $1.361$ (for $n \hspace{-2pt} = \hspace{-2pt} p \hspace{-2pt} = \hspace{-2pt} 1$)                                         
to $0.153$ (for $n \hspace{-2pt} = \hspace{-2pt} p \hspace{-2pt} = \hspace{-2pt}2$).                                                                             
This nine-fold reduction for second-derivative estimation reflects the steeper VRF decay ($N^{-5}$ vs.\ $N^{-3}$),                                               
which results in more precise \sure{} cost estimates and thus requires less regularization.         

Algorithm~\ref{alg:surede_soft} summarizes the soft-combining \surede{} procedure.  The differences from 
Algorithm~\ref{alg:surede_hard} are: (i)~the temperature $T^*$ is pre-computed from~\eqref{eq:T_star_explicit}; (ii)~the
output is the weighted average~\eqref{eq:soft_estimate} rather than a single selected estimate.
The computational cost per sample is identical to hard-combining. The temperature $T^*$ and all filter quantities are pre-computed,
adding no per-sample overhead.

\begin{algorithm}[t]
\caption{Soft-combining \surede{}}
\label{alg:surede_soft}
\begin{algorithmic}[1]
\Require A set $\mathcal{N} \hspace{-2pt} = \hspace{-2pt} \{N_1 \hspace{-2pt} < \hspace{-2pt} \cdots \hspace{-2pt} < \hspace{-2pt} N_M\}$;
       polynomial degree~$p$; derivative order~$n$;
       noise covariance $\mathsf{Q}_w$; parameter~$N_0$.
\Statex \emph{Initialization (once):}
\State Pre-compute $\mathbf{v}(N)$, $\mathbf{s}^{(n)}_{N_0}$, $\mathbf{v}(N)^T \mathsf{Q}_w^{(N)} \mathbf{s}^{(n)}_{N_0}$, 
      $\forall N \hspace{-2pt} \in \hspace{-2pt} \mathcal{N}$.
\State Compute $T^*$ from~\eqref{eq:T_star_explicit}
     using $N = N_1$.
\Statex \emph{Per sample $k$:}
\For{each $N \in \mathcal{N}$}
\State $e(N) \gets \mathbf{v}(N)^T \mathbf{y}_k(N)$
       \Comment{derivative estimate}
\State $r(N) \gets \bigl(\mathbf{s}^{(n)}_{N_0}\bigr)^T
       \mathbf{y}_k(N)$ \Comment{s-vector projection}
\State $\mathfrak{c}(N) \hspace{-2pt} \gets \hspace{-2pt} N_0 e(N)^2
          +  2\,\tau(N)
         \hspace{-2pt} - \hspace{-2pt} 2e(N) r(N)$      
\EndFor
\State $w(N) \gets
     \dfrac{N \cdot \exp\!\bigl(-\mathfrak{c}(N) / T^*\bigr)}
           {\sum_{M} M \cdot
            \exp\!\bigl(-\mathfrak{c}(M) / T^*\bigr)}$,
     \quad  $\forall N \in \mathcal{N}$
\State \Return $\hat{x}^{(n)}_{\mathrm{soft}}
     = \sum_{N} w(N)\,e(N)$
     \Comment{soft-combined estimate}
\end{algorithmic}
\end{algorithm}

\vspace{-3mm}
\section{Numerical Results}
\label{sec:numerical_results}

\subsection{Simulation Setup}
\label{sec:sim_setup}

The test signal comprises three 200-sample segments (periods 15, 40, and 100) to exercise short, medium, 
and long filter regimes. Observations are corrupted by Gaussian noise $\sigma \hspace{-2pt} \in \hspace{-2pt} \{0.005, 0.05, 0.15\}$, 
with results averaged over $M \hspace{-2pt} = \hspace{-2pt} 500$ Monte Carlo trials. We evaluate first-derivative estimation ($n=1$) using $p=1$ 
for all LS-based methods. Adaptive methods (\awve{}, \ici{} with $\Gamma \hspace{-2pt} = \hspace{-2pt} 2.0$, and \surede{}) optimize over 
$\mathcal{N} \hspace{-2pt} = \hspace{-2pt} \{4, 8, \dots, 24\}$, while fixed-length baselines use $N \hspace{-2pt} \in \hspace{-2pt} \{2, 4, 8, 24\}$.
Comparison baselines include a constant-velocity Kalman filter (process noise $\sigma_Q \hspace{-2pt} = \hspace{-2pt} \sigma_R$) 
and non-causal Savitzky--Golay filters (SG-11, SG-21, $p \hspace{-2pt} = \hspace{-2pt} 2$). 
Note that the SG filters use future samples, providing an inherent information advantage. 
To ensure fairness, neither the Kalman filter nor \surede{} uses scenario-specific tuning. 
The overall MSE is the mean of the three segments, excluding the first 50 samples. 
Reported differences exceeding 10\% are statistically significant ($p < 0.05$).

\vspace{-3mm}
\subsection{MSE Comparison}
\label{sec:mse_comparison}

Figure~\ref{fig:mse_comparison} shows the per-sample $\mse{}$ across all three noise levels.
In the fast sinusoid segment, short filters achieve low bias but high variance, 
while long filters suffer from smoothing-induced bias; 
all three adaptive methods select shorter windows and achieve $\mse{}$ close to the best fixed filter.
In the medium and slow sinusoid segments, longer filters become advantageous and 
the adaptive methods correctly extend their effective window lengths.
\surede{}'s soft-combining consistently yields the lowest $\mse{}$ among the adaptive 
methods across all segments and noise levels. 
A bias-variance decomposition (not shown) indicates that the adaptive methods achieve near-optimal trade-offs.
Their bias remains close to that of short filters in the fast segment, while their variance approaches that of long filters in the slow segments.

\begin{figure*}[t]
\centering
\includegraphics[width=0.9\textwidth]{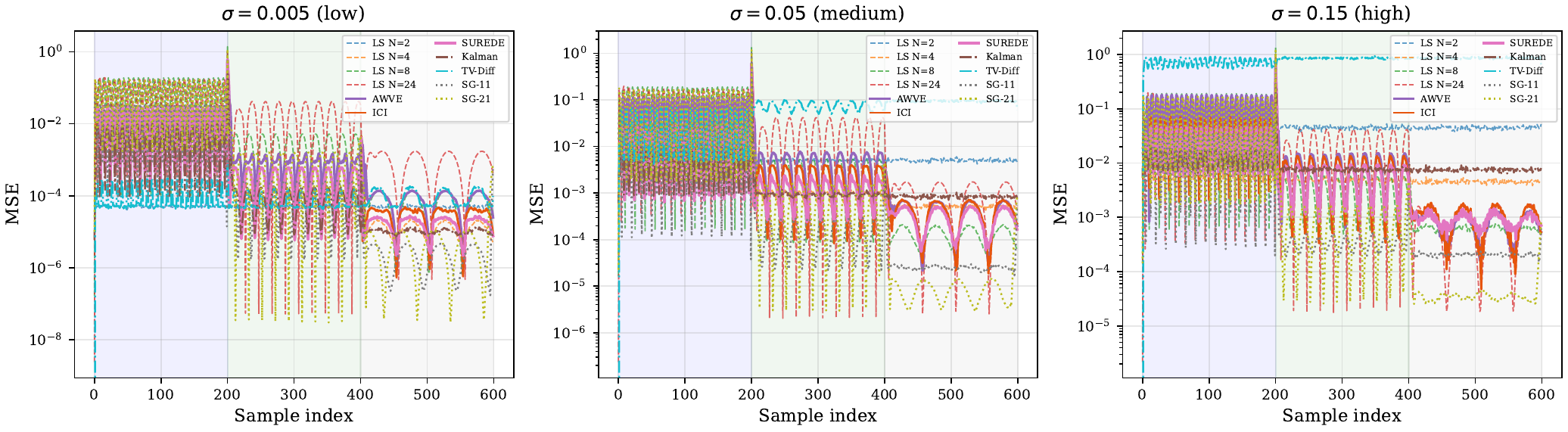}
\caption{Per-sample $\mse{}$ comparison across noise levels ($\sigma \hspace{-2pt} \in \hspace{-2pt} \{0.005, 0.05, 0.15\}$). 
Solid lines denote adaptive methods (\awve{}, \ici{}, \surede{}), 
dashed lines denote fixed-length LS filters, and dash-dot lines indicate alternative baselines. 
Background shading distinguishes the three signal segments.}
\vspace{-4mm}
\label{fig:mse_comparison}
\end{figure*}

Table~\ref{tab:mse_medium} summarizes segment-averaged $\mse{}$ for the medium noise regime ($\sigma=0.05$). 
\surede{} achieves the lowest overall $\mse{}$ among adaptive methods ($5.52 \times 10^{-3}$), 
outperforming \ici{} and \awve{} by 8\% and $4\times$, respectively. 
While the Kalman filter shows a localized advantage and SG-11 remains competitive, 
these behaviors shift at higher noise levels.

\begin{table}[t]
\centering
\caption{Mean $\mse{}$ ($\times 10^{-3}$) per segment, $\sigma = 0.05$. Best causal results in \textbf{bold}.}
\label{tab:mse_medium}
\addtolength{\tabcolsep}{-3pt}
\small 
\begin{tabular}{lcccc}
\toprule
Method & Fast & Med. & Slow & Overall \\
\midrule
LS $N\!=\!2$   & $\mathbf{5.02}$ & $5.01$ & $5.05$ & $5.03$ \\
LS $N\!=\!4$   & $14.9$ & $\mathbf{0.81}$ & $0.51$ & $5.42$ \\
LS $N\!=\!8$   & $95.8$ & $2.75$ & $\mathbf{0.13}$ & $32.9$ \\
LS $N\!=\!24$  & $85.7$ & $21.4$ & $0.86$ & $36.0$ \\
\midrule
\awve{}        & $60.9$ & $4.34$ & $0.38$ & $21.9$ \\
\ici{}         & $15.3$ & $2.33$ & $0.40$ & $6.00$ \\
\surede{}      & $15.0$ & $1.26$ & $0.30$ & $5.52$ \\
\midrule
Kalman         & $5.95$ & $0.92$ & $0.84$ & $\mathbf{2.57}$ \\
SG-11 (n.c.)   & $18.4$ & $0.17$ & $0.03$ & $6.20$ \\
SG-21 (n.c.)   & $82.9$ & $0.84$ & $0.02$ & $27.9$ \\
\bottomrule
\end{tabular}
\end{table}

Table~\ref{tab:mse_all_noise} compares overall $\mse{}$ across noise levels. 
\surede{} consistently leads the adaptive methods, outperforming \ici{} by margins that grow from 
negligible at low noise to 25\% at $\sigma=0.15$. 
In contrast, \awve{}'s performance degrades nearly fourfold over the same range. 
While the Kalman filter's temporal prediction provides an advantage at lower noise, 
it is surpassed by \surede{} in the high-noise regime. 
Fixed-length LS and SG filters remain relatively noise-insensitive, 
as their error is dominated by bias, rendering them uncompetitive against adaptive selection as noise increases.

\begin{table}[t]
\centering
\caption{Overall $\mse{}$ ($\times 10^{-3}$) across noise levels ($n = 1$). Best causal results in \textbf{bold}.}
\label{tab:mse_all_noise}
\addtolength{\tabcolsep}{-2pt} 
\small
\begin{tabular}{lccc}
\toprule
Method & $\sigma = 0.005$ & $\sigma = 0.05$ & $\sigma = 0.15$ \\
\midrule
LS $N\!=\!2$   & $0.05$ & $5.03$ & $45.0$ \\
LS $N\!=\!4$   & $4.91$ & $5.42$ & $9.35$ \\
LS $N\!=\!8$   & $32.8$ & $32.9$ & $33.3$ \\
LS $N\!=\!24$  & $36.0$ & $36.0$ & $36.0$ \\
\midrule
\awve{}        & $5.24$ & $21.9$ & $32.2$ \\
\ici{}         & $4.92$ & $6.00$ & $12.1$ \\
\surede{}      & $4.91$ & $5.52$ & $\mathbf{9.18}$ \\
\midrule
Kalman         & $\mathbf{1.74}$ & $\mathbf{2.57}$ & $9.18$ \\
SG-11 (n.c.)   & $6.18$ & $6.20$ & $6.39$ \\
SG-21 (n.c.)   & $27.9$ & $27.9$ & $28.0$ \\
\bottomrule
\end{tabular}
\end{table}

\vspace{-3mm}
\subsection{Robustness to Noise Variance Misspecification}
\label{sec:noise_misspec}

All three adaptive methods require estimate of $\sigma^2$.
In practice, this is rarely known exactly, so robustness to misspecification is critical.
Figure~\ref{fig:noise_sensitivity} and shows the overall $\mse{}$ when the assumed noise 
standard deviation~$\hat\sigma$ deviates from the true value 
($\sigma \hspace{-2pt} = \hspace{-2pt} 0.05$) by factors ranging from $0.5\times$ to $2.0\times$.

\begin{figure}[t]
\centering
\includegraphics[width=0.90\linewidth]{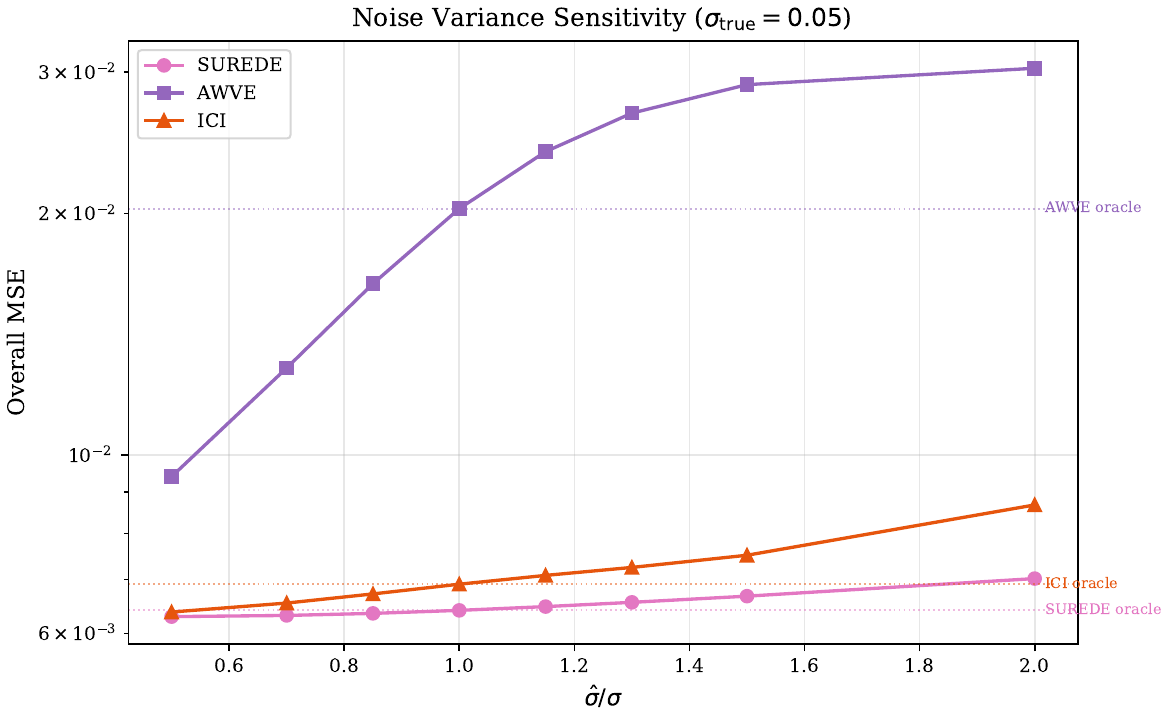}
\caption{Overall $\mse{}$ as a function of the noise misspecification
ratio $\hat\sigma / \sigma$.}
\label{fig:noise_sensitivity}
\end{figure}

\surede{} exhibits superior robustness, degrading by only 12\% over a $4\times$ variance 
misspecification range ($\hat\sigma/\sigma \in [0.5, 2.0]$), 
compared to 36\% for \ici{} and 222\% for \awve{}. 
This stability stems from the soft-combining mechanism: blending candidate filters prevents 
the catastrophic misselection common in the binary, threshold-sensitive decisions of \awve{} and \ici{}.

For practical implementation, the noise variance is estimated via the Median Absolute Deviation (MAD) 
of first differences $\mathbf{d}$ as $\hat\sigma = 1.4826 \cdot \mathrm{MAD}(\mathbf{d})$ \cite{donoho1994ideal}. 
This estimator is robust to outliers and can be applied over a rolling window to track non-stationary noise.

\vspace{-3mm}
\subsection{An Alternative Signal}
\label{sec:generalization}

To ensure the results are signal-independent, we evaluated a second test signal comprising a chirp (varying $N^*$), 
a ramp-to-quadratic transition (accelerating dynamics), and a quintic polynomial (high-order smoothness).
Table~\ref{tab:second_signal_mse} confirms that \surede{} remains the most effective adaptive method, 
particularly in the chirp segment where it matches the performance of the shortest fixed filter (LS~$N=2$) 
while avoiding the bias seen in longer filters. 
While the Kalman filter's constant-velocity model yields the lowest chirp error, 
it lacks the flexibility to handle the non-constant velocity of the ramp and quintic segments. 
Conversely, fixed SG filters fail to generalize, suffering severe bias in the ramp-to-quadratic segment. 
\surede{}'s competitive performance on the quintic segment ($9.30 \times 10^{-5}$) further highlights the advantage of soft-combining for smooth transitions.

\begin{table}[t]
\centering
\caption{Mean $\mse{}$ ($\times 10^{-3}$) for secondary test signal at $\sigma = 0.05$. Best causal results in \textbf{bold}.}
\label{tab:second_signal_mse}
\addtolength{\tabcolsep}{-3pt}
\small
\begin{tabular}{lcccc}
\toprule
Method & Chirp & Ramp-Quad & Quintic & Overall \\
\midrule
LS $N\!=\!2$   & $5.04$ & $4.99$ & $5.02$ & $5.02$ \\
LS $N\!=\!24$  & $57.7$ & $\mathbf{0.01}$ & $0.08$ & $19.3$ \\
\midrule
\awve{}        & $25.6$ & $0.04$ & $\mathbf{0.07}$ & $8.57$ \\
\ici{}         & $5.71$ & $0.03$ & $0.09$ & $1.94$ \\
\surede{}      & $5.02$ & $0.07$ & $0.09$ & $1.73$ \\
\midrule
Kalman         & $\mathbf{2.25}$ & $0.83$ & $0.84$ & $\mathbf{1.31}$ \\
SG-11 (n.c.)   & $6.02$ & $13.8$ & $0.02$ & $6.62$ \\
SG-21 (n.c.)   & $26.5$ & $7.49$ & $0.01$ & $11.3$ \\
\bottomrule
\vspace{-4mm}
\end{tabular}
\end{table}

\begin{figure}[t]
\centering
\includegraphics[width=0.90\linewidth]{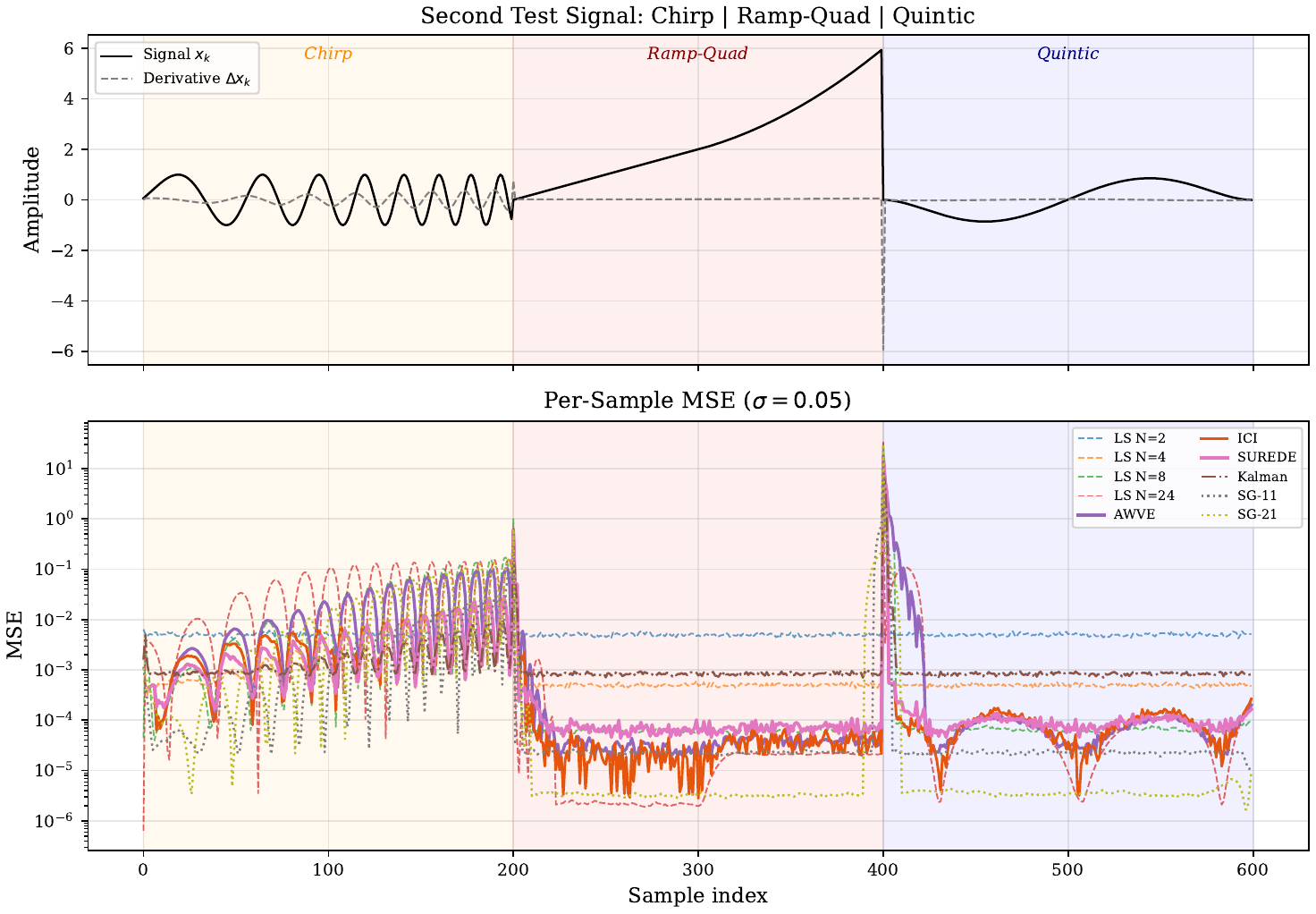}
\caption{Second test signal. Top: clean signal and derivative.
Bottom: per-sample $\mse{}$ ($\sigma = 0.05$).  \surede{} tracks the
varying dynamics across all three segments.}
\label{fig:second_signal_mse}
\vspace{-4mm}
\end{figure}

\vspace{-3mm}
\subsection{Discussion of Savitzky--Golay Performance}
\label{sec:sg_honest}

SG-11’s comparable performance at medium noise ($6.20 \hspace{-2pt} \times \hspace{-2pt} 10^{-3}$) stems from three non-representative advantages: 
(i) an \emph{information advantage} via non-causality (5-sample look-ahead), 
(ii) superior local bias properties from a higher polynomial order 
($p \hspace{-2pt} = \hspace{-2pt} 2$ vs. $p \hspace{-2pt} = \hspace{-2pt} 1$), and 
(iii) a window length ($w \hspace{-2pt} = \hspace{-2pt} 11$) that happens to be near-optimal for this specific signal.
However, these benefits are fragile. Unlike \surede{}, which is inherently causal and adapts its window at each 
sample via the Theorem~\ref{thm:oracle_ineq} oracle, SG filters require manual tuning of both $w$ and $p$. 
Poor parameter selection, seen in SG-21's fivefold higher MSE, and fixed-window constraints lead to substantial 
performance degradation in non-stationary settings. 
This is confirmed by the EuRoC real-data experiments (see Section~\ref{sec:real_data}).

\vspace{-4mm}
\section{EuRoC MAV Case Study}
\label{sec:real_data}
\vspace{-1mm}

To validate \surede{} beyond synthetic signals, we evaluate all methods
on position-to-velocity differentiation using the EuRoC MAV
dataset~\cite{burri2016euroc}, which provides Vicon-tracked
trajectories at ${\sim}200$\,Hz with sub-millimeter accuracy.

\vspace{-3mm}
\subsection{Experimental Protocol}
\label{sec:euroc_setup}

We validate the proposed framework on the EuRoC MAV dataset \cite{burri2016euroc} using two sequences: 
\textbf{MH\_01\_easy} (smooth flight) and \textbf{V1\_02\_medium} (aggressive maneuvers). 
Ground truth velocity is obtained via central differences on the Vicon position signal, 
which is corrupted by i.i.d. Gaussian noise $\sigma \hspace{-2pt} \in \hspace{-2pt} \{2, 5, 10, 50\}$,mm to 
simulate tracking uncertainties ranging from high-end optical systems to integrated IMU drift \cite{merriaux2017study}. 

To maintain a viable per-sample signal-to-noise ratio (SNR) for risk discrimination, 
we down-sample the native 200,Hz signal by a factor of 6 ($\approx$33,Hz). 
This aligns the real-world SNR with our synthetic experiments ($SNR \approx 6$ at $\sigma=5$,mm) and 
follows standard practice for reducing high-frequency noise amplification in derivative estimation \cite{brown1992algorithms}.

Consistent with our synthetic study, all adaptive methods optimize over 
$\mathcal{N} \hspace{-2pt} = \hspace{-2pt} \{4, 8, \dots, 24\}$ with $p \hspace{-2pt} = \hspace{-2pt} 1$, 
while fixed LS baselines use $N \hspace{-2pt}\in \hspace{-2pt} \{4, 8, 16, 24\}$. 
The Kalman filter employs a constant-velocity model with process noise $\sigma_Q \hspace{-2pt} = \hspace{-2pt} \sigma_R$. 
No per-scenario tuning was performed for any method, ensuring a fair comparison across all flight dynamics.

\vspace{-3mm}
\subsection{Results and Discussion}
\label{sec:euroc_results_dis}

Tables~\ref{tab:euroc_mh01} and~\ref{tab:euroc_v102} report per-sequence RMSE across four noise levels. 
At $\sigma \hspace{-2pt}  \le \hspace{-2pt}  10$\,mm, \surede{} consistently achieves the lowest RMSE among adaptive methods, 
outperforming \ici{} and \awve{} by up to 16\% and 41\%, respectively. 
Figure~\ref{fig:euroc_velocity} illustrates this advantage: while short filters (LS~$N \hspace{-2pt} = \hspace{-2pt} 4$) amplify noise 
and long filters (LS~$N=24$) lag during maneuvers, \surede{} adapts its window locally to track transients smoothly.

The Kalman filter, despite its temporal prediction advantage, collapses as noise increases; 
at $\sigma \hspace{-2pt} = \hspace{-2pt} 10$\,mm, its RMSE is nearly double that of \surede{}. 
This suggests that the fixed bandwidth of the $\sigma_Q \hspace{-2pt} = \hspace{-2pt} \sigma_R$
Kalman model cannot accommodate the mismatch between constant-velocity assumptions and aggressive flight dynamics. 
At extreme noise ($\sigma \hspace{-2pt} = \hspace{-2pt} 50$\,mm), long-window methods (LS~$N \hspace{-2pt} = \hspace{-2pt} 16$, SG-21) dominate as 
the optimal window shifts beyond the adaptive range, with \ici{}'s hard selection providing better 
regularization than \surede{}'s soft blending in this high-variance regime.

Consistent with synthetic results, \surede{} remains remarkably robust to noise misspecification on real data. 
Over a $4\times$ range of $\hat{\sigma}/\sigma$, its RMSE varies by only 14\%–19\%, whereas \ici{} and \awve{} degrade
by up to 70\% and 89\%, respectively. 

\begin{table}[t]
\centering
\caption{RMSE (m/s) on MH\_01\_easy. Best causal in \textbf{bold}.}
\label{tab:euroc_mh01}
\addtolength{\tabcolsep}{-4pt}
\small
\begin{tabular}{lcccc}
\toprule
Method & 2\,mm & 5\,mm & 10\,mm & 50\,mm \\
\midrule
LS $N\!=\!4$   & 0.054 & 0.086 & 0.157 & 0.752 \\
LS $N\!=\!8$   & 0.068 & \textbf{0.072} & \textbf{0.085} & 0.263 \\
LS $N\!=\!16$  & 0.116 & 0.116 & 0.117 & \textbf{0.145} \\
LS $N\!=\!24$  & 0.164 & 0.164 & 0.164 & 0.170 \\
\midrule
\awve{}        & 0.077 & 0.094 & 0.113 & 0.216 \\
\ici{}         & 0.061 & 0.082 & 0.106 & 0.199 \\
\surede{}      & \textbf{0.054} & 0.074 & 0.099 & 0.281 \\
\midrule
Kalman (CV)    & 0.056 & 0.103 & 0.198 & 0.963 \\
SG-11 (n.c.)   & 0.105 & 0.106 & 0.109 & 0.190 \\
SG-21 (n.c.)   & 0.080 & 0.080 & 0.082 & 0.099 \\
\bottomrule
\end{tabular}
\end{table}

\begin{table}[t]
\centering
\caption{RMSE (m/s) on V1\_02\_medium. Best causal in \textbf{bold}.}
\label{tab:euroc_v102}
\addtolength{\tabcolsep}{-4pt}
\small
\begin{tabular}{lcccc}
\toprule
Method & 2\,mm & 5\,mm & 10\,mm & 50\,mm \\
\midrule
LS $N\!=\!4$   & \textbf{0.060} & \textbf{0.091} & 0.158 & 0.727 \\
LS $N\!=\!8$   & 0.116 & 0.118 & 0.124 & 0.280 \\
LS $N\!=\!16$  & 0.211 & 0.211 & 0.210 & \textbf{0.230} \\
LS $N\!=\!24$  & 0.267 & 0.267 & 0.266 & 0.272 \\
\midrule
\awve{}        & 0.112 & 0.138 & 0.160 & 0.264 \\
\ici{}         & 0.075 & 0.105 & 0.146 & 0.258 \\
\surede{}      & 0.066 & 0.092 & \textbf{0.122} & 0.312 \\
\midrule
Kalman (CV)    & 0.053 & 0.103 & 0.196 & 0.944 \\
SG-11 (n.c.)   & 0.153 & 0.155 & 0.157 & 0.222 \\
SG-21 (n.c.)   & 0.129 & 0.130 & 0.131 & 0.143 \\
\bottomrule
\end{tabular}
\vspace{-4mm}
\end{table}

\begin{figure}[t]
\centering
\includegraphics[width=\linewidth]{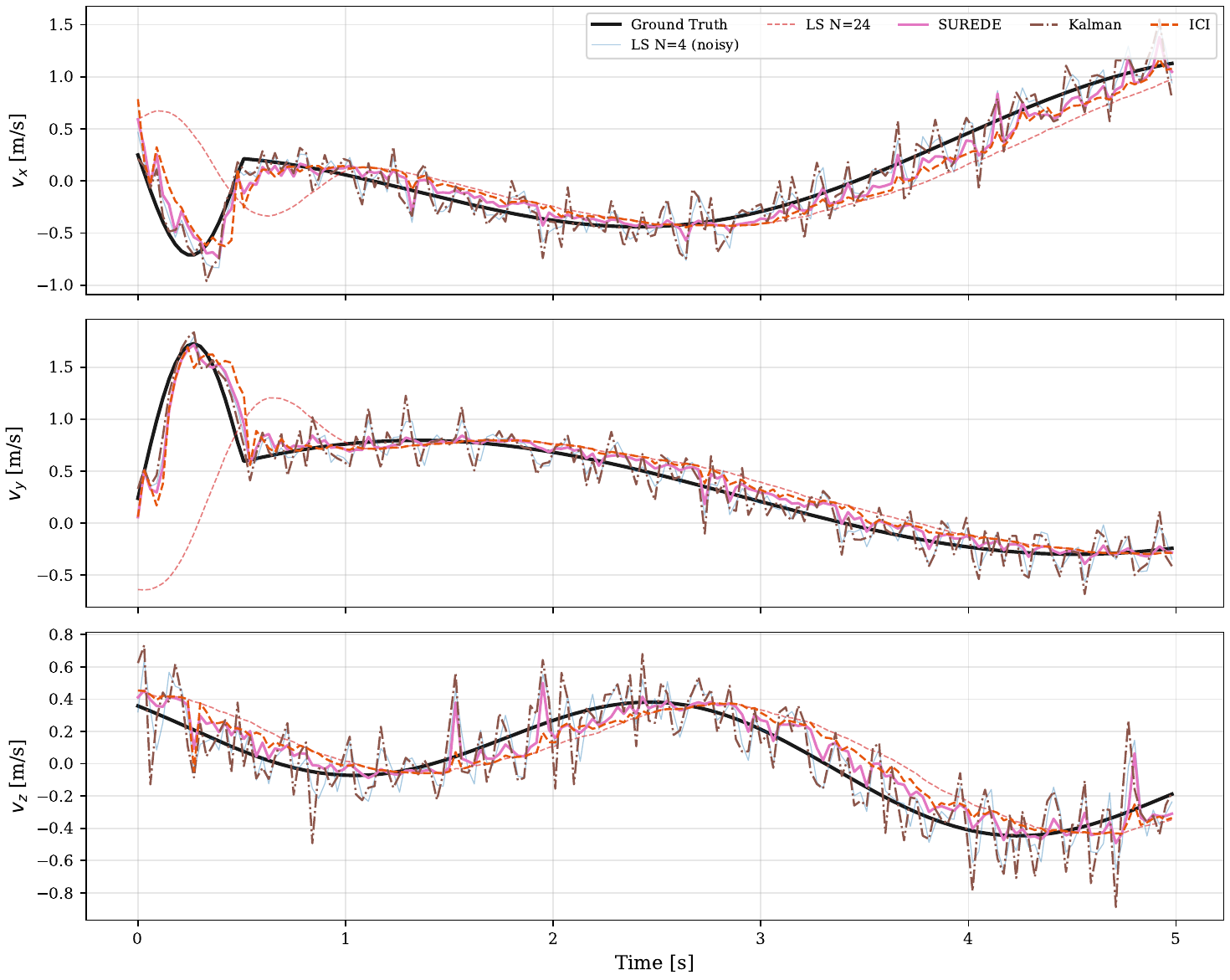}
\caption{Velocity estimates on V1\_02\_medium ($\sigma = 10$\,mm). \surede{} (magenta) balances 
the noise-sensitivity of short windows (blue) and the lag of long windows (red dashed).}
\label{fig:euroc_velocity}
\vspace{-5mm}
\end{figure}

\vspace{-3mm}
\section{Conclusion}
\label{sec:conclusion}

We presented \surede{}, a tuning-free framework for causal derivative estimation that minimizes a 
closed-form SURE cost to select or soft-combine a bank of FIR filters. 
Theoretically, we established that \surede{} satisfies a minimax-optimal $O(\sqrt{\log K})$ oracle inequality. 
Experimentally, \surede{} consistently outperformed adaptive methods (\ici{}, \awve{}) and parametric baselines (Kalman filter) 
across synthetic and real-world EuRoC MAV datasets at low-to-moderate noise. 
Notably, \surede{} exhibited superior robustness to noise misspecification, with MSE degrading by only 16\% 
over a fourfold variance range—significantly less than the 70\%--144\% seen in competing methods. 
While nonparametric FIR approaches are inherently limited at extreme noise levels compared to 
tuned parametric models, the general \surede{} framework remains applicable to any linear estimator and derivative order. 
Future work will focus on online noise calibration, extension to event-driven sampling, and integration 
with adaptive smoothing filters like the 1\euro{} filter.


\vspace{-4mm}
\bibliographystyle{IEEEtran}
\bibliography{references}

\appendix
\numberwithin{equation}{section}
\vspace{-3mm}
\section{Derivation of the \surede{} Cost Function}
\label{app:derivation}

We derive the \sure{} cost \eqref{eq:sure_cost} by mapping the derivative estimation problem 
to the general SURE framework \cite{eldar2009gsure}. 
Starting from the observation model $\mathbf{y} \hspace{-2pt} = \hspace{-2pt} \mathbf{x} \hspace{-2pt} + \hspace{-2pt} \mathbf{w}$, 
we apply the $n$-th order derivative filter matrix $\mathsf{H}_d^{(n)}$ 
to obtain the transformed model 
$\tilde{\mathbf{y}} \hspace{-2pt} = \hspace{-2pt} \mathbf{x}^{(n)} \hspace{-2pt} + \hspace{-2pt} \tilde{\mathbf{w}}$, 
where $\tilde{\mathbf{y}} \hspace{-2pt} = \hspace{-2pt} \mathsf{H}_d^{(n)}\mathbf{y}$, 
$\mathbf{x}^{(n)} \hspace{-2pt} = \hspace{-2pt} \mathsf{H}_d^{(n)}\mathbf{x}$, 
and $\tilde{\mathbf{w}} \hspace{-2pt} \sim \hspace{-2pt} \mathcal{N}(\mathbf{0}, \tilde{\mathsf{Q}})$ 
with $\tilde{\mathsf{Q}} \hspace{-2pt} = \hspace{-2pt} \mathsf{H}_d^{(n)}\mathsf{Q}_w^{(N)}(\mathsf{H}_d^{(n)})^T$. 
In this setting, the ML estimator is $\hat{\boldsymbol{\theta}}_{\mathrm{ML}} \hspace{-2pt} = \hspace{-2pt} \tilde{\mathbf{y}}$.

The estimator $\mathbf{h}(\mathbf{u}) \hspace{-2pt} = \hspace{-2pt} \mathsf{R}\tilde{\mathbf{y}}$ 
produces $N_0$ simultaneous derivative estimates, where $\mathsf{R}$ applies the LS filter $\mathbf{v}^T$. 
Under the constant-derivative approximation, the SURE terms are evaluated as follows:
\begin{itemize}
    \item \textbf{Estimation Norm:} $\|\mathbf{h}(\mathbf{u})\|^2 \hspace{-2pt} = \hspace{-2pt} N_0 (\mathbf{v}^T \mathbf{y})^2$.
    \item \textbf{Jacobian Trace:} Since $\mathbf{h}$ is linear, 
    the Jacobian $\partial\mathbf{h}/\partial\mathbf{u} \hspace{-2pt} = \hspace{-2pt} \mathsf{R}\tilde{\mathsf{Q}}$. 
    Its trace reduces to $\operatorname{Tr}(\partial\mathbf{h}/\partial\mathbf{u}) \hspace{-2pt} = \hspace{-2pt} \mathbf{v}^T \mathsf{Q}_w^{(N)}\mathbf{s}^{(n)}_{N_0}$, 
    where the $s$-vector aggregates observation contributions across the $N_0$ outputs.
    \item \textbf{Cross Term:} The term $\mathbf{h}(\mathbf{u})^T\hat{\boldsymbol{\theta}}_{\mathrm{ML}}$ simplifies to 
    $(\mathbf{v}^T \mathbf{y}) (\mathbf{s}^{(n)}_{N_0})^T \mathbf{y}$ by utilizing the 
    identity $\mathbf{1}_{N_0}^T \mathsf{H}_d^{(n)} \mathbf{y} \hspace{-2pt} = \hspace{-2pt} (\mathbf{s}^{(n)}_{N_0})^T \mathbf{y}$.
\end{itemize}

Substituting these into the general SURE expression 
$S(N) \hspace{-2pt} = \hspace{-2pt} \|\boldsymbol{\theta}\|^2 \hspace{-1pt} + \hspace{-1pt} \|\mathbf{h}(\mathbf{u})\|^2 \hspace{-1pt} + \hspace{-1pt} 2\operatorname{Tr}(\frac{\partial\mathbf{h}}{\partial\mathbf{u}}) \hspace{-1pt} - \hspace{-1pt} 2\mathbf{h}(\mathbf{u})^T\hat{\boldsymbol{\theta}}_{\mathrm{ML}}$ 
and omitting $\|\boldsymbol{\theta}\|^2$ (which is $N$-independent) yields the \surede{} cost function \eqref{eq:sure_cost}. \qed


\vspace{-4mm}
\section{Proof of Thm.~\ref{thm:oracle_ineq}}
\label{app:oracle_proofs}


\vspace{-2mm}
\subsection{Variance of the \surede{} Cost}
\label{sec:sure_variance}

To establish the oracle inequality, we control the fluctuations 
$\xi(N) \hspace{-2pt} = \hspace{-2pt} \mathfrak{c}(N) \hspace{-2pt} - \hspace{-2pt} \mathbb{E}[\mathfrak{c}(N)]$. 
Substituting the Gaussian decompositions 
$\mathbf{v}^T \mathbf{y} \hspace{-2pt} = \hspace{-2pt} \mu_v \hspace{-2pt} + \hspace{-2pt} Z_v$ 
and $\mathbf{s}^T \mathbf{y} \hspace{-2pt} = \hspace{-2pt} \mu_s \hspace{-2pt} + \hspace{-2pt} Z_s$ into \eqref{eq:sure_cost} 
(where $\mu_v \hspace{-2pt} = \hspace{-2pt} \mathbf{v}^T \mathbf{x}$ and $\mu_s \hspace{-2pt} = \hspace{-2pt} \mathbf{s}^T \mathbf{x}$ 
are deterministic signal components, 
and $Z_v \hspace{-2pt} = \hspace{-2pt} \mathbf{v}^T \mathbf{w}$ and $Z_s \hspace{-2pt} = \hspace{-2pt} \mathbf{s}^T \mathbf{w}$ are zero-mean Gaussian noise terms), 
the fluctuation partitions into orthogonal linear and quadratic components ($\xi_1, \xi_2$) 
belonging to the first and second Wiener chaoses~\cite{nualart2006malliavin}:
\vspace{-1mm}
\begin{align}
    \xi(N) &= \underbrace{(2N_0\mu_v - 2\mu_s)Z_v - 2\mu_v Z_s}_{\xi_1} \nonumber \\
    & \quad+ \underbrace{N_0(Z_v^2 - V_w) - 2(Z_vZ_s - C_{vs})}_{\xi_2}.
    \label{eq:xi_decomp_short}
\end{align}
Since $\mathbb{E}[\xi_1\xi_2] \hspace{-2pt} = \hspace{-2pt} 0$, the total 
variance is $\text{Var}[\xi] \hspace{-2pt} = \hspace{-2pt} \text{Var}[\xi_1] \hspace{-2pt} + \hspace{-2pt} \text{Var}[\xi_2]$. 
This decomposition ensures that fluctuations are bounded by the filter's bias and noise variance, 
providing the concentration necessary for model selection.

\begin{lemma}[Variance bound for the \surede{} cost]
  \label{lem:sure_variance}
  Under Gaussian noise, the variance of the fluctuation satisfies:
  \begin{equation}
    \text{Var}\bigl[\xi(N)\bigr] \hspace{-2pt} \le \hspace{-2pt} C_1\bigl(b^2(N) \hspace{-2pt} + \hspace{-2pt} 1\bigr)V_{\text{eff}}(N) \hspace{-2pt} + \hspace{-2pt} C_2 V_{\text{eff}}^2(N),
    \label{eq:sure_var_bound_short}
  \end{equation}
  where $V_{\text{eff}}(N)$ defined in \eqref{eq:Veff_def}, $b(N)$ is the filter bias, 
  $C_1$ depends on the signal derivative $x^{(n)}_k$, and $C_2 = 2N_0^2 + 8N_0 + 8$.
\end{lemma}

\begin{proof}
  \textbf{Linear Part:} $\xi_1 \hspace{-2pt} = \hspace{-2pt} a_1 Z_v \hspace{-2pt} + \hspace{-2pt} a_2 Z_s$ 
  is a weighted sum of Gaussian variables with weights $a_1 \hspace{-2pt} = \hspace{-2pt} 2N_0\mu_v \hspace{-2pt} - \hspace{-2pt} 2\mu_s$ 
  and $a_2 \hspace{-2pt} = \hspace{-2pt} -2\mu_v$. 
  By Cauchy-Schwarz, $\text{Var}[\xi_1] \hspace{-2pt} \le \hspace{-2pt} (|a_1| \hspace{-2pt} + \hspace{-2pt} |a_2|)^2 V_{\text{eff}}$. 
  Given $\mu_v \hspace{-2pt} = \hspace{-2pt} b(N) \hspace{-2pt} + \hspace{-2pt} x^{(n)}_k$ and 
  $\mu_s \hspace{-2pt} \hspace{-2pt} = \hspace{-2pt} \hspace{-2pt} \mathbf{s}^T\mathbf{x}$, 
  the weights satisfy 
  $(|a_1| \hspace{-2pt} + \hspace{-2pt} |a_2|)^2 \hspace{-2pt} \le \hspace{-2pt} 6(N_0 \hspace{-1pt} + \hspace{-1pt}1)^2 b^2(N) \hspace{-1pt} + \hspace{-1pt} 3\alpha^2$, 
  where $\alpha$ aggregates all terms that are independent of $N$ (specifically $x^{(n)}_k$ and $N_0$).
  This yields $\text{Var}[\xi_1] \hspace{-2pt} \le \hspace{-2pt} C_1 (b^2(N) \hspace{-2pt} + \hspace{-2pt} 1) V_{\text{eff}}$.

  \textbf{Quadratic Part:} Using the fourth-moment formula for jointly Gaussian variables, 
  $\text{Var}[\xi_2] \hspace{-2pt} = \hspace{-2pt} 2N_0^2 V_w^2 \hspace{-1pt} - \hspace{-1pt} 8N_0 V_w C_{vs} \hspace{-1pt} + \hspace{-1pt} 4V_w S_w \hspace{-1pt} + \hspace{-1pt} 4C_{vs}^2$. 
  Applying the bound $|C_{vs}| \hspace{-2pt} \le \hspace{-2pt} V_{\text{eff}}$ (by Cauchy--Schwarz), 
  we obtain $\text{Var}[\xi_2] \hspace{-2pt} \le \hspace{-2pt} (2N_0^2 \hspace{-1pt} + \hspace{-1pt} 8N_0 \hspace{-1pt} + \hspace{-1pt} 8)V_{\text{eff}}^2$. 
  
  Combining both parts concludes the proof.
\end{proof}

\vspace{-4mm}
\subsection{Concentration of \surede{} Fluctuations}
\label{sec:concentration}

To prove the oracle inequality (Thm.~\ref{thm:oracle_ineq}), we must bound the tail probabilities of the cost fluctuations $\xi(N)$. 
Since $\xi(N)$ is an at-most-degree-2 polynomial in the Gaussian vector $\mathbf{w}$, 
it can be expressed as a sum of a quadratic form and a linear term: 
$\xi(N) \hspace{-2pt} = \hspace{-2pt} \mathbf{w}^T \mathsf{M}(N) \mathbf{w} \hspace{-1pt} + \hspace{-1pt} \mathbf{m}(N)^T \mathbf{w} \hspace{-1pt} - \hspace{-1pt} \mathbb{E}[\mathbf{w}^T\mathsf{M}(N)\mathbf{w}]$,
Here $\mathsf{M}(N)$ and $\mathbf{m}(N)$ are symmetric matrix and vector
that depend on the signal~$\mathbf{x}$, the filter $\mathbf{v}(N)$, and vector $\mathbf{s}$.

\begin{lemma}[Hanson--Wright Concentration]
  \label{lem:concentration}
  $\forall N \hspace{-2pt} \in \hspace{-2pt} \mathcal{N}$, there exists a constant $c_1 \hspace{-2pt} > \hspace{-2pt} 0$ such that $\forall t \hspace{-2pt} > \hspace{-2pt} 0$:
  \vspace{-1mm}
  \begin{align}
    \Pr\bigl(|\xi(N)| > t\bigr) & \le \nonumber \\
    & \hspace{-45pt} 2\exp\left( -c_1 \min\left( \frac{t^2}{\text{Var}[\xi(N)]}, \frac{t}{V_{\text{eff}}(N)} \right) \right).
    \label{eq:concentration}
  \end{align}
\end{lemma}

\begin{proof}
  The result follows from the Hanson--Wright inequality for quadratic forms of Gaussian vectors \cite{rudelson2013hansonwright}. 
  The sub-Gaussian tail ($t^2$ term) dominates for moderate deviations, while the sub-exponential tail ($t$ term) governs large deviations. 
  Here, $\text{Var}[\xi(N)]$ controls the fluctuations in the small-$t$ regime, 
  while $V_{\text{eff}}$ scales with the spectral norm of $\mathsf{M}(N)$, 
  providing the necessary concentration for model selection.
\end{proof}

\vspace{-4mm}
\subsection{Proof of Theorem \ref{thm:oracle_ineq}}
\label{sec:thm2_proof}

We follow a three-step concentration argument. 

\textbf{1) Basic Inequality:} By the optimality of $\hat{N}$ with respect to the cost $\mathfrak{c}(N)$, 
we have $\mathfrak{c}(\hat{N}) \hspace{-2pt} \le \hspace{-2pt} \mathfrak{c}(N^*)$. 
Expanding this using the decomposition \eqref{eq:sure_decomposition} yields:
\begin{equation}
    R(\hat{N}) \hspace{-2pt}  \le \hspace{-2pt} R^* \hspace{-2pt} + \hspace{-2pt} \xi(N^*) \hspace{-2pt} - \hspace{-2pt} \xi(\hat{N}) \le R^* \hspace{-2pt} + \hspace{-2pt} 2\max_{N \in \mathcal{N}} |\xi(N)|.
    \label{eq:proof_step2}
\end{equation}

\textbf{2) Concentration and Union Bound:} We control the term $\max_N |\xi(N)|$ by applying Lemma~\ref{lem:concentration} with a union bound over the $K$ candidates in $\mathcal{N}$. For any $t > 0$:
\vspace{-1mm}
\begin{equation}
    \Pr\left( \max_{N \in \mathcal{N}} |\xi(N)| > t \right) \le 2K e^{ -c_1 \min\left( \frac{t^2}{\overline{V}^2}, \frac{t}{\overline{V}} \right) },
    \label{eq:union_bound}
\end{equation}
where $\overline{V}$ bounds the variance terms from Lemma~\ref{lem:sure_variance}. 
Setting the right-hand side to $\delta$ and solving for $t$ confirms that the fluctuations scale 
as $\overline{V}\sqrt{\log(K/\delta)}$, leading to the high-probability bound \eqref{eq:oracle_ineq_prob}.

\textbf{3) Expectation Bound:} Finally, integrating the tail probability via $\mathbb{E}[\max |\xi|] = \int_0^\infty \Pr(\max|\xi| > t)\,dt$ 
and noting that the sub-Gaussian tail dominates for moderate $t$ establishes the expected risk bound \eqref{eq:oracle_ineq_expect}. 
This confirms that the SURE-selected estimator tracks the oracle risk to within an optimal $O(\sqrt{\log K})$ factor. 

\vspace{-3mm}
\section{Proof of Thm. \ref{thm:optimal_temp}}
\label{app:opt_temp_proof}

Minimizing the bound~\eqref{eq:ewa_bound} with respect to $T$, 
yields the following expression:
\begin{equation}
    T^{\ast} = \sqrt{\frac{\max_{N \in \mathcal{N}} \Var[\xi(N)] }{4 \log K}}.
    \label{eq:opt_T}
\end{equation}
The cost fluctuation $\xi(N)$ decomposes into a linear part
$\xi_1$ (signal-dependent) and a quadratic part $\xi_2$
(signal-independent), which are orthogonal by the Wiener chaos
structure, see~\eqref{eq:xi_decomp_short}.
In the balanced regime it holds that $\Var[\xi(N)] \hspace{-2pt} = \hspace{-2pt} 2\Var[\xi_2(N)]$,
and $\text{Var}[\xi_2]$ is maximized at the shortest candidate window $N_1$. 
Plugging the expression for $\text{Var}[\xi_2]$ from the proof of Lemma \ref{lem:sure_variance},
evaluated at the shortest candidate window $N_1$, into~\eqref{eq:opt_T} leads to~\eqref{eq:T_star_explicit}.
 

\end{document}